\documentclass[journal]{IEEEtran}
\usepackage{graphicx}
\usepackage{amsfonts,bm}
\usepackage{array}
\usepackage{booktabs}
\usepackage{multirow}
\usepackage{threeparttable}
\usepackage{diagbox}
\usepackage{algorithm}
\usepackage[noend]{algpseudocode}
\usepackage{silence}
\usepackage{subcaption}
\usepackage{color}
\usepackage[dvipsnames]{xcolor}

\usepackage{amsmath,amssymb,amsthm}
\newtheorem{theorem}{Theorem}
\newtheorem{corollary}{Corollary}

\usepackage[hyphens]{url}
\usepackage{cite}
\usepackage{hyperref}
\usepackage{cleveref}
\makeatletter
\renewcommand{\ALG@beginalgorithmic}{\footnotesize}
\makeatother

\begin{document}
\pagenumbering{arabic} 
\title{Amplitude-Independent Robust Snapshot 6-D Radio SLAM via a Unified Angle-Delay Formulation}
\author{
Shengqiang~Shen,~\IEEEmembership{Member,~IEEE}, Aoyun~Hao, Weihao~Geng, Lei~Yang, Shiyin~Li, and~Henk Wymeersch,~\IEEEmembership{Fellow,~IEEE}
}
\maketitle
\begin{abstract}
This paper addresses bistatic snapshot radio SLAM, in which a user equipment (UE) with unknown 6-D pose and clock bias is localized and environmental landmarks are reconstructed from a single multipath channel snapshot. Under mixed line-of-sight (LoS)/non-line-of-sight (NLoS) propagation, existing robust snapshot SLAM methods are mainly developed or validated in planar/2-D settings and often use path-amplitude or path-loss information for LoS handling, which makes them sensitive to calibration errors and propagation-model mismatch. We propose an amplitude-independent robust radio SLAM method built on a unified angle-delay formulation for LoS and single-bounce NLoS inlier paths. In the coarse stage, the method estimates the UE state and selects geometrically consistent inliers directly from angle-delay measurements, without amplitude-based LoS preclassification or path-wise latent variables; the formulation is further extended to general 3-D/6-D pose estimation through twist-swing two-stage traversal initialization and local refinement on $SO(3)$. A subsequent Jacobian-row-equilibrated iteratively reweighted least-squares (IRLS) refinement, combined with quasi-Akaike information criterion (QAIC) model comparison, detects the LoS path and jointly refines the UE state and scattering points. We also analyze formulation-specific local-rank properties and their minimal-set implications under unknown path identity. Simulations show that the proposed method remains competitive with calibrated amplitude-dependent baselines and is more robust to path-loss-model mismatch.
\end{abstract}

\section{Introduction}\label{Introduction} 

\IEEEPARstart{T}{he} ongoing rollout of 5G and evolution toward 6G enable high-precision positioning and sensing for emerging applications \cite{9976205}, \cite{9330512}. These capabilities are driven by mmWave spectrum and MIMO architectures, which provide wide bandwidth and large antenna arrays for improved delay and angular resolution \cite{9957135}, \cite{10818978}. Resolvable multipath components are particularly useful because their delay and angular parameters carry both UE-state information and environmental landmark information, enabling radio SLAM, including single-BS bistatic radio SLAM in cellular mmWave settings \cite{amjad2023radio,10694114,10556695}, which jointly estimates the user equipment (UE) pose, clock bias, and environmental scattering points. 

In radio SLAM, the line-of-sight (LoS) path and single-bounce reflections are typically treated as geometric inliers because their delay/angle measurements are well captured by simple propagation models. Specifically, the LoS path directly relates the transmitter and receiver states, while a single-bounce path can be represented through one dominant scattering point, such as a wall reflection; both path types therefore provide useful constraints for localization and mapping. In contrast, multi-bounce propagation often deviates from the single-scatter model and is commonly treated as an outlier, since it can exhibit longer delays and more complex angle-delay relationships. Following \cite{10509551}, we refer to single-bounce paths as NLoS-1 and multi-bounce paths as NLoS-$n$ in this paper, as illustrated in Fig. \ref{SLAM_paradigm}.

\begin{figure}
    \centering
    \includegraphics[width=0.9\linewidth]{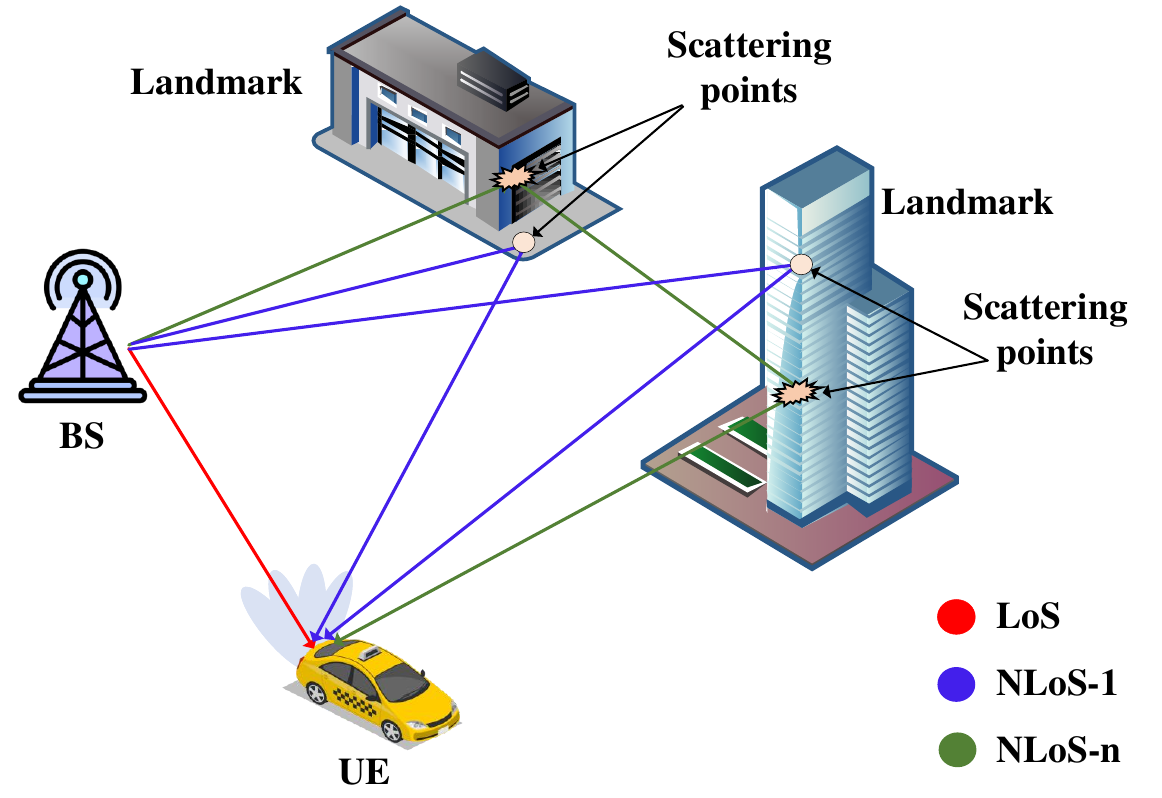}
    \caption{\small Illustration of the bistatic radio SLAM paradigm where UE jointly estimates its state as well as those of the environment landmarks.}
    \label{SLAM_paradigm}
\end{figure}

Within bistatic radio SLAM, snapshot radio SLAM is designed on a single channel snapshot without prior user states or motion models \cite{10694114}, \cite{9367007}; it is therefore useful as a standalone baseline and as an initializer for sequential schemes \cite{9724223,9179819,10568571}. Beyond snapshot settings, sequential multipath-based SLAM has also been developed to fuse information across multiple propagation paths and radio-reflective surfaces \cite{leitinger2023data}. An inherent challenge of snapshot radio SLAM is the ambiguity of path identity, since LoS, NLoS-1, and NLoS-$n$ paths have distinct geometric characteristics and require different modeling assumptions. Early snapshot radio SLAM formulations assume known path identity and inlier-only measurements \cite{9367007,9179819,10424688}. More recent robust snapshot methods, mainly developed or validated for planar/2-D settings, further handle outliers and unknown path identity through robust estimation or Bayesian hypothesis handling \cite{10818978,11046111}. However, these robust snapshot baselines often rely on path-amplitude or path-loss information for LoS handling, making their performance dependent on calibration quality and propagation-model accuracy. Moreover, although both LoS and NLoS-1 paths are valid inliers, confusing these two path types can bias the state estimate because the LoS path has no associated scattering point \cite{11366027}, \cite{HU20254457}.

To address these challenges, we develop an amplitude-independent robust snapshot radio SLAM method based on a unified angle-delay formulation. The proposed formulation uses angle-delay geometric consistency to handle LoS and NLoS-1 inliers, reject NLoS-$n$ outliers, and extend naturally to general 3-D/6-D pose estimation without introducing path-wise latent variables. Built on this formulation, the coarse stage performs robust initialization and path-consistency testing directly from the UE state and angle-delay observations, avoiding amplitude-based LoS preclassification. LoS detection is deferred to a subsequent model-selection refinement stage, thereby reducing the impact of LoS/NLoS-1 misclassification within the inlier set. 
Our main contributions are summarized as follows:
\begin{itemize}
    \item We propose a unified angle-delay coarse-stage formulation for LoS and NLoS-1 inlier paths. The formulation is written directly in the UE state and channel angle-delay observations, avoids path-wise latent variables in the coarse-stage solver state, and does not require amplitude-based LoS preclassification.
    \item We develop an amplitude-independent robust initialization pipeline under unknown path identity and NLoS-n outliers. The pipeline embeds the proposed coarse solver into a consensus-based minimal-set search with geometric feasibility checks, while formulation-specific local-rank analysis provides the corresponding minimal-set implications.
    \item We extend the formulation to a general 3-D setting for 6-D pose and clock-bias estimation, with the planar case treated as a special case of the same vector construction. For the 3-D case, the position and clock bias are updated in closed form for a given orientation, while the UE orientation is initialized by twist-swing two-stage traversal and then refined locally on $SO(3)$.
    \item Building on the coarse estimate, we develop a Jacobian-row-equilibrated IRLS refinement stage with QAIC-based model comparison that jointly refines the UE state, detects the LoS path, and estimates scattering points without relying on calibrated path-loss information.
\end{itemize}

The rest of the paper is organized as follows. Section \ref{Problem formulation and model} presents the problem formulation and system models. Section \ref{Proposed Amplitude-Independent Method} develops the proposed angle-delay formulation, robust initialization, and refinement scheme. Section \ref{Formulation-Specific Local-Rank and Minimal-Sample Implications} discusses local-rank and minimal-sample implications. Section \ref{Results} reports simulation results, and Section \ref{Conclusion} concludes the paper.
\begin{figure}[t]
  \centering
  \includegraphics[width=0.95\linewidth]{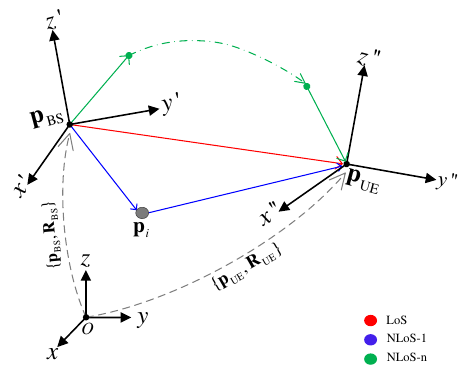} 
  \caption{\small A spatial schematic that depicts multiple signal propagation paths with one LoS path (red line), one NLoS-1 path (blue line), and one NLoS-n path (green line). The local frame of the BS $ (x',y',z')$ and that of UE $(x'',y'',z'')$, respectively, has relative transformations of $\{\mathbf{p}_{\rm{BS}},\mathbf{R}_{\rm{BS}}\}$ and $\{\mathbf{p}_{\rm{UE}},\mathbf{R}_{\rm{UE}}\}$ with respect to the global coordinate system $(x, y,z)$. The point ${{\bf{p}}_{i}}$ denotes the position of the scattering point of the NLoS-1 path. }
  \label{geometry_3d}
\end{figure}

\section{Problem formulation and models}\label{Problem formulation and model} 

We consider a bistatic downlink snapshot radio SLAM problem with stationary BS and UE. The BS state, including position $\mathbf{p}_{\rm{BS}}\in\mathbb{R}^3$ and orientation $\mathbf{R}_{\rm{BS}}\in SO(3)$, is assumed known, while the UE state is $\mathbf{x}=\{\mathbf{p}_{\rm{UE}},b_{\rm{UE}},\mathbf{R}_{\rm{UE}}\}$, where position $\mathbf{p}_{\rm{UE}}\in\mathbb{R}^3$, clock bias $b_{\rm{UE}}\in\mathbb{R}$, and orientation $\mathbf{R}_{\rm{UE}}\in SO(3)$. As shown in Fig.~\ref{geometry_3d}, the propagation paths are classified as LoS, single-scattering NLoS (NLoS-1), and multi-scattering NLoS (NLoS-n). An NLoS-1 path is associated with a scattering point $\mathbf{p}_i\in\mathbb{R}^3$.

Radio SLAM is typically performed in two phases: multipath channel-parameter estimation followed by geometric localization and mapping \cite{10818978}, \cite{8240645}. This work focuses on the second phase and treats the estimated path parameters as observations.

\subsection {Geometric Measurement Models}\label{Geometric Measurement Models}
We assume that antenna arrays are available at both the BS and UE, and that a standard channel-estimation method \cite{9976205}, \cite{9724223}, \cite{9620874} provides the AoD $\bm{\phi}_i=[{\phi _{a,i}},{\phi _{e,i}}]^ \top$, AoA $\bm{\theta }_i = [{\theta _{a,i}},{\theta _{e,i}}]^ \top$, delay ${\tau _i}$, Doppler frequency shift ${\nu _i}$, and complex gain ${\xi _i}$ of the $i$th path \cite{10568571}, \cite{9591345}. Since UE mobility is not considered, we set $\nu_i=0$ or absorb it into the complex gain.

The geometric measurement model uses these channel estimates as observations. Let $c$ denote the speed of light, let ${{\bf{z}}_i} = [c{\tilde \tau _i},\boldsymbol{\tilde \phi _{i}},\boldsymbol{\tilde \theta _{i}}]$ denote the measurement of the $i$th path, and let $\mathcal{Z}$ and $\mathcal{I}$ denote the measurement set and associated indices. Assuming zero-mean Gaussian measurement noise \cite{10818978}, the likelihood is
\begin{equation}\label{observe pdf}
p({{\bf{z}}_i}\mid {\bf{\Theta}}_i) = {\cal N}({{\bf{z}}_i}\mid {\bf{h}}_i({\bf{\Theta}}_i),{{\bf{\Sigma}}_i}),
\end{equation}
where ${{\bf{h}}_i} = [c{\tau}_{i}, \boldsymbol{\phi}_{i}, \boldsymbol{\theta}_{i}]$ is the noiseless channel-parameter vector parameterized by ${\bf{\Theta}}_i$, and ${{\bf{\Sigma}}_i}$ is the measurement covariance. Here $\mathbf{\Sigma}_i$ describes the data-generating uncertainty; the proposed estimator does not assume calibrated path-wise covariance and therefore uses the covariance-free Jacobian-row-equilibrated IRLS refinement in Section~\ref{Iterative Estimation and LoS Detection}. The mean model differs for LoS, NLoS-1, and NLoS-n paths.

\textbf{LoS Path:} The LoS mean model \cite{9500778} is 
\begin{equation}\label{observe mean LoS}
\left\{ {\begin{array}{*{20}{l}}
\begin{array}{l}
c{\tau _{{\rm{LoS}}}} = \left\| {{{\bf{p}}_{{\rm{BS}}}} - {{\bf{p}}_{{\rm{UE}}}}} \right\| + c{b_{{\rm{UE}}}}\\
{\phi _{{\rm{a}},{\rm{LoS}}}} = {\tan ^{ - 1}}\left( {\frac{{{\bf{e}}_2^ \top {{\bf{R}}^ \top_{{\rm{BS}}}}\left( {{{\bf{p}}_{{\rm{UE}}}} - {{\bf{p}}_{{\rm{BS}}}}} \right)}}{{{\bf{e}}_1^ \top {{\bf{R}}^ \top_{{\rm{BS}}}}\left( {{{\bf{p}}_{{\rm{UE}}}} - {{\bf{p}}_{{\rm{BS}}}}} \right)}}} \right)\\
{\phi _{{\rm{e}},{\rm{LoS}}}} = \sin ^{ - 1} \left( {\frac{{{\bf{e}}_3^ \top {{\bf{R}}^ \top_{{\rm{BS}}}}\left( {{{\bf{p}}_{{\rm{UE}}}} - {{\bf{p}}_{{\rm{BS}}}}} \right)}}{{\left\| {{{\bf{p}}_{{\rm{BS}}}} - {{\bf{p}}_{{\rm{UE}}}}} \right\|}}} \right)
\end{array}\\
\begin{array}{l}
{\theta _{{\rm{a,LoS}}}} = {\tan ^{ - 1}}\left( {\frac{{{\bf{e}}_2^ \top {{\bf{R}}^ \top_{{\rm{UE}}}}\left( {{{\bf{p}}_{{\rm{BS}}}} - {{\bf{p}}_{{\rm{UE}}}}} \right)}}{{{\bf{e}}_1^ \top {{\bf{R}}^ \top_{{\rm{UE}}}}\left( {{{\bf{p}}_{{\rm{BS}}}} - {{\bf{p}}_{{\rm{UE}}}}} \right)}}} \right)\\
{\theta _{{\rm{e,LoS}}}} = \sin ^{ - 1} \left( {\frac{{{\bf{e}}_3^ \top {{\bf{R}}^ \top_{{\rm{UE}}}}\left( {{{\bf{p}}_{{\rm{BS}}}} - {{\bf{p}}_{{\rm{UE}}}}} \right)}}{{\left\| {{{\bf{p}}_{{\rm{BS}}}} - {{\bf{p}}_{{\rm{UE}}}}} \right\|}}} \right)
\end{array}
\end{array}} \right.,
\end{equation}
where $\mathbf{e}_k$ denotes the $k$th canonical basis vector. \eqref{observe mean LoS} implies that ${\bf{\Theta}}_{i} \doteq \{{\mathbf{p}}_{\rm{UE}},b_{\rm{UE}},\mathbf{R}_{\rm{UE}}\}\in \mathbb{R}^{4}\times SO(3)$.

\textbf{NLoS-1 Path:} The NLoS-1 mean model \cite{9500778} is 
\begin{equation}\label{observe mean}
\left\{ {\begin{array}{*{20}{l}}
\begin{array}{l}
{c{\tau _i} = \left( {\left\| {{{\bf{p}}_{{\rm{BS}}}} - {{\bf{p}}_i}} \right\| + \left\| {{{\bf{p}}_i} - {{\bf{p}}_{{\rm{UE}}}}} \right\|} \right) + c{b_{{\rm{UE}}}}}\\
{\phi _{a,i}} = {\tan ^{ - 1}}\left( {\frac{{{\bf{e}}_2^ \top {{\bf{R}}^ \top_{{\rm{BS}}}}\left( {{{\bf{p}}_i} - {{\bf{p}}_{{\rm{BS}}}}} \right)}}{{{\bf{e}}_1^ \top {{\bf{R}}^ \top_{{\rm{BS}}}}\left( {{{\bf{p}}_i} - {{\bf{p}}_{{\rm{BS}}}}} \right)}}} \right)\\
{\phi _{{\rm{e,i}}}} = \sin ^{ - 1} \left( {\frac{{{\bf{e}}_3^ \top {{\bf{R}}^ \top_{{\rm{BS}}}}\left( {{{\bf{p}}_i} - {{\bf{p}}_{{\rm{BS}}}}} \right)}}{{\left\| {{{\bf{p}}_{{\rm{BS}}}} - {{\bf{p}}_i}} \right\|}}} \right)\\
{\theta _{a,i}} = {\tan ^{ - 1}}\left( {\frac{{{\bf{e}}_2^\top {{\bf{R}}^ \top_{{\rm{UE}}}} \left( {{{\bf{p}}_i} - {{\bf{p}}_{{\rm{UE}}}}} \right)}}{{{\bf{e}}_1^\top {{\bf{R}}^ \top_{{\rm{UE}}}}\left( {{{\bf{p}}_i} - {{\bf{p}}_{{\rm{UE}}}}} \right)}}} \right)\\
{{\theta _{{\rm{e,i}}}} = \sin ^{ - 1} \left( {\frac{{{\bf{e}}_3^\top {{\bf{R}}^ \top_{{\rm{UE}}}} \left( {{{\bf{p}}_i} - {{\bf{p}}_{{\rm{UE}}}}} \right)}}{{\left\| {{{\bf{p}}_i} - {{\bf{p}}_{{\rm{UE}}}}} \right\|}}} \right)}
\end{array}
\end{array}} \right.,
\end{equation}
Thus, ${\bf{\Theta}}_{i} \doteq \{{\mathbf{p}}_{\rm{UE}},b_{\rm{UE}},{\mathbf{p}}_{i},\mathbf{R}_{\rm{UE}}\}\in \mathbb{R}^{7}\times SO(3)$ for an NLoS-1 path.

\textbf{NLoS-n Path:} Since NLoS-n paths do not admit a simple single-scatter geometric model, they are described directly by ${{\bf{\Theta }}_i} \buildrel\textstyle.\over= \{ c{\tau _i},\bm{\phi} _{i},\bm{\theta} _{i}\}\in\mathbb{R}^{5} $.

This paper focuses on estimating the UE state and NLoS-1 scattering points while identifying LoS/NLoS-1 inliers and rejecting NLoS-n outliers from the measurement set $\mathcal{Z}$. The unknown path identity motivates the robust formulation developed next.
\section{Proposed Amplitude-Independent Method }\label{Proposed Amplitude-Independent Method}

This section develops the proposed unified angle-delay formulation for LoS and NLoS-1 inlier paths. We first present a planar special case embedded in the same 3-D vector construction, then extend it to general 3-D/6-D pose estimation, unknown path identity, outlier rejection, and Jacobian-row-equilibrated IRLS refinement with QAIC-based model comparison.

\subsection{Amplitude-Independent Method in the absence of Outliers}
\label{Amplitude-Independent Optimization Method Without Outliers}
We begin with the planar special case because it exposes the geometric structure with lighter notation. The purpose is to derive an angle-delay expression of the UE state without introducing the unknown incidence/scattering point locations $\mathbf p_i$ as variables. The resulting equations provide a novel amplitude-independent coarse-stage formulation, written directly in the UE state and angle-delay observations, and lead naturally to the general 3-D case.

\subsubsection{\textbf{Planar Scenario}} 
In this scenario, all scatterers lie in a common horizontal plane, and the elevation components of AoA and AoD are disregarded. As shown in Fig. \ref{geometry_2dto3d_framework}, the planar construction is illustrated as a top-view schematic viewed along the global $z$ axis, with the global $X$-$Y$ coordinates indicated in the figure. The planar case is treated as a special case of the general 3-D vector construction. The UE and BS orientations are described by the yaw angles ${\alpha_{\rm{UE}}}$ and ${\alpha_{\rm{BS}}}$, with $\mathbf{R}_{\rm{UE}} = \rm{blkdiag}\left(\tilde{\mathbf{R}}(\alpha_{\rm{UE}}),1\right)$ and $\mathbf{R}_{\rm{BS}} = \rm{blkdiag}\left(\tilde{\mathbf{R}}(\alpha_{\rm{BS}}),1\right)$, where
\begin{equation}\label{eq:planar_yaw_rotation}
\tilde{\mathbf{R}}(\alpha ) = \begin{bmatrix}
\cos (\alpha ) & - \sin (\alpha ) \\
\sin (\alpha ) & \cos (\alpha ) 
\end{bmatrix} \in SO(2),
\end{equation}
and the AoA $\theta$ and AoD $\phi$ are expressed in the local frames of the UE and BS, respectively.

\begin{figure}[t]
  \centering
  \includegraphics[width=0.85\linewidth]{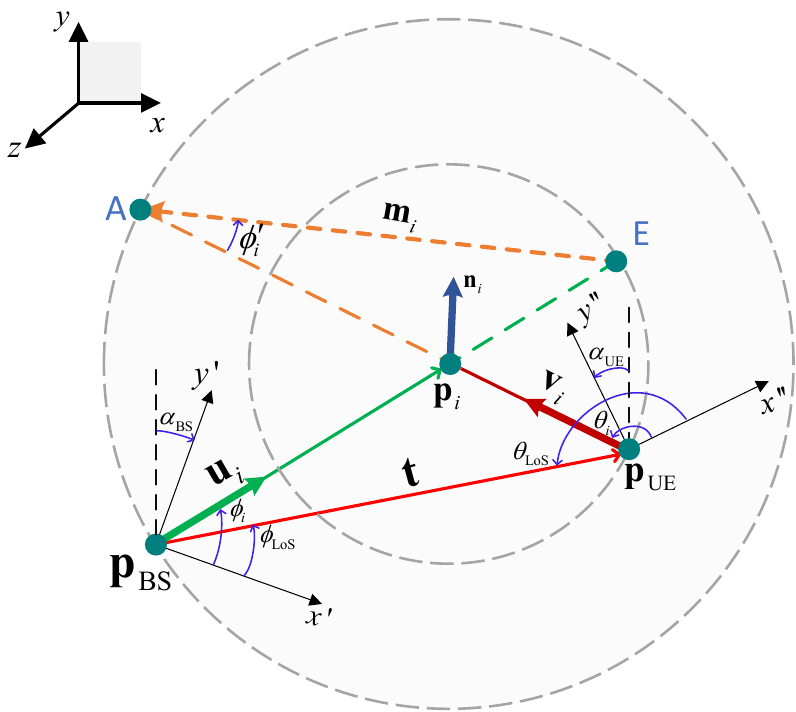} 
  \caption{\small Top-view planar schematic of the angle-delay construction, viewed along the global $z$ axis and projected onto the $X-Y$ plane. The global $X$-$Y$ coordinates are indicated in the figure. For illustration, one LoS path and one NLoS-1 path are shown. Auxiliary circles are introduced so that the angle-delay consistency of a valid NLoS-1 path can be expressed directly in the UE state. The derivation itself does not assume that the LoS identity is known a priori. The scattering point ${{\bf{p}}_{{i}}}$ shares the center of the gray dashed circles that intersect $\mathbf{p}_{\rm{BS}}$ and $\mathbf{p}_{\rm{UE}}$, respectively.}
  \label{geometry_2dto3d_framework}
\end{figure}

In Fig. \ref{geometry_2dto3d_framework}, two auxiliary circles centered at ${{\bf{p}}_{i}}$ are used to express the NLoS-1 angle-delay consistency directly in the UE state. Defined in the global coordinate system, the AoD and AoA direction vectors of the $i$th path are
\begin{equation}\label{doa unit}
{{\bf{u}}_i} = {\bf{R}}_{\rm{BS}}{\left[ 
{\cos ({\phi _i})}, {\sin ({\phi _i})}, 0\right]^ \top },
\end{equation}
 \begin{equation}\label{aoa unit}
{{\bf{v}}_i} = {\bf{R}}_{\rm{UE}}{\left[ {\cos ({\theta _i})}, {\sin ({\theta _i})},0 \right]^ \top }.
\end{equation}

Let ${{\bf{t}}}= {{\bf{p}}_{\rm{UE}}} - {{\bf{p}}_{\rm{BS}}}$ denote the BS-to-UE displacement, and introduce ${{\bf{m}}_i}$ as an auxiliary vector. From the symmetry of the auxiliary construction about the scattering point $\mathbf{p}_i$, the corresponding vector relations satisfy
\begin{equation}\label{solving 1}
\left\{
\begin{array}{c}
{\bf{t}} \times {{\bf{v}}_i} = {{\bf{u}}_i} \times {{\bf{m}}_i}  \\
\mathbf{t}^{\top}\mathbf{v}_i = \mathbf{u}_i^{\top}\mathbf{m}_i.
\end{array}
\right.
\end{equation}
Here, $\mathsf{AE}$ denotes the auxiliary segment from $\mathsf{E}$ to $\mathsf{A}$, with $\mathbf{m}_i=\overrightarrow{\mathsf{E}\mathsf{A}}$. By the symmetry of the auxiliary construction, $\mathsf{AE}$ is congruent to the BS-to-UE segment, and hence $\|\mathbf{m}_i\|_2=\|\mathbf{t}\|_2$. Equivalently,
\begin{equation}\label{solving 2}
{{\bf{m}}_i} = \underbrace{({\bf{t}} + {d_i}{{\bf{v}}_i})}_{\rm{Point}\; \mathsf{A}} - \underbrace{{d_i}{{\bf{u}}_i}}_{\rm{Point}\; \mathsf{E}},
\end{equation}
and $d_i = c(\tau_i - b_{\rm{UE}})$ is the propagation distance of the $i$th path after clock-bias correction. The auxiliary segment {\sf{AE}} allows the path constraint to be written in terms of the UE state and angle-delay observations, without introducing an extra path-wise variable such as $\mathbf{p}_i$ or the fractional parameter in \cite{9179819}. The same relation covers both LoS and NLoS-1 inlier paths, with LoS as a degenerate case.

Substituting \eqref{solving 2} into \eqref{solving 1} and using the vector triple-product identity yields the compact constraint
\begin{equation}\label{solving 3}
\mathbf{M}_i\mathbf{t}=d_i(\mathbf{v}_i-\mathbf{u}_i),
\end{equation}
where
\begin{equation}\label{M}
    \mathbf{M}_i = \mathbf{v}_i \mathbf{u}_i^\top+\mathbf{u}_i \mathbf{v}_i^\top-(\mathbf{u}_i^\top \mathbf{v}_i+1)\mathbf{I}.
\end{equation}
This constraint is equivalent to \eqref{solving 1}--\eqref{solving 2}. Matrix $\mathbf M_i$ maps $\mathbf t$ onto the path-induced constraint subspace determined by $(\mathbf u_i,\mathbf v_i)$; for an NLoS-1 path, its nullspace is along $\mathbf n_i=\mathbf u_i+\mathbf v_i$, while LoS is a more degenerate case.

For each path, $d_i$ depends on the clock bias and $\mathbf{v}_i$ depends on the UE orientation. Thus, \eqref{solving 3} defines the residual
\begin{equation}\label{solving 4}
\mathbf{r}_i(\mathbf{x})\doteq 
\mathbf{M}_i\mathbf{t}-
{d_i(\mathbf{v}_i-\mathbf{u}_i)}.
\end{equation}
where ${\bf{x}} = {[{\bf{ p}}_{{\rm{UE}}}^ \top ,{ b_{{\rm{UE}}}},{ \alpha _{{\rm{UE}}}}]^ \top }$, which is free of the scattering point positions $\mathbf{p}_i$. With noisy observations, we estimate $\mathbf{x}$ by LS over the inlier set $\mathcal I$:
\begin{equation}\label{cost 1}
J({\bf{x}}) = \textstyle \sum_{i \in {\mathcal I}} J_i({\bf{x}}),\quad J_i({\bf{x}}) =  \| \mathbf{r}_i({\bf{x}})\|^2,
\end{equation}
and
\begin{equation}\label{minLS}
\hat{\mathbf{x}} = {\arg\min _{{\mathbf{x}}}} J({\bf{x}}).
\end{equation}
For a fixed heading ${\alpha _{{\rm{UE}}}}$, the position and clock bias admit a closed-form update. We therefore search over $\alpha_{\rm{UE}}\in\mathcal A$ and choose the heading that minimizes \eqref{minLS}, while all candidate inlier paths are evaluated by the same angle-delay constraint without LoS preclassification.

\textbf{Position and Clock Bias Estimation:}\label{Position And clock bias Estimation}
For a given ${\alpha _{{\rm{UE}}}}$,  \eqref{cost 1} can be written as 
\begin{equation} \label{cost 2}
J({\bf{x}}) = 
\Bigg\| 
\mathbf{b} - \mathbf{A}
\begin{bmatrix}
\mathbf{p}_{\rm{UE}} - \mathbf{p}_{\rm{BS}}\\
-cb_{\rm{UE}}
\end{bmatrix}
\Bigg\|^2,
\end{equation}
with
\begin{equation}\label{AA}
{\bf{A}} = \left[
    \cdots,
    {\mathbf{A}}_i^{\top},
    \cdots,
    \left[0,0,1,0\right]^{\top}
\right]^{\top} \in {\mathbb{R}}^{(3|\mathcal{I}|+1)\times 4},
\end{equation}
and
\begin{equation}\label{bb}
{\bf{b}} = [\cdots,{{\mathbf{b}}_{i}^ \top},\cdots,0]^ \top  \in {\mathbb{R}}^{(3|\mathcal{I}|+1)\times 1},
\end{equation}
for $i\in\mathcal{I}$. In \eqref{AA} the submatrix 
\begin{equation}
{{\bf{A}}_i} =\left[\mathbf{M}_i\quad-\left(\mathbf{v}_i-\mathbf{u}_i\right)\right]\in\mathbb{R}^{3\times 4}
\end{equation}
while in \eqref{bb} the subvector ${{\bf{b}}_i} = c{\tau_i}({{\bf{v}}_i} - {{\bf{u}}_i})$.  
Since we discuss the planar scenario and solve the problem in a 3-D view, we set the last row of $\bf{A}$ and $\bf{b}$ to be $\left[0,0,1,0\right]$ and 0, respectively, to constrain the third dimension of the estimate to be 0, ensuring that the final solution is confined to the horizontal plane. 
Then, a closed-form solution can be obtained as
\begin{equation}\label{solving}
\left[({\hat {\bf{p}}_{{\rm{UE}}}} - {{\bf{p}}_{{\rm{BS}}}})^ \top, -c\hat{b}_{\rm{UE}}\right]^ \top = \mathbf{A}^{\dagger}{\mathbf{b}}.
\end{equation} 
where $(\cdot)^{\dagger}$ computes the pseudoinverse.

\textbf{UE Orientation Estimation:}\label{UE Orientation Estimation
}
For each candidate $\alpha\in\mathcal A$, let $\hat{\mathbf{x}}_{\alpha}\doteq[\hat{\mathbf p}_{\rm UE}^{\top}(\alpha),\hat b_{\rm UE}(\alpha),\alpha]^{\top}$ denote the state obtained by the closed-form position and clock-bias update. The optimal $\alpha_{\rm{UE}}$ is then obtained by 
\begin{equation}\label{aue}
\hat{\alpha}_{\rm UE} = \arg\min_{\alpha\in\mathcal A} J(\hat{\mathbf{x}}_{\alpha}).
\end{equation}
Here $J(\hat{\mathbf{x}}_{\alpha})$ denotes $J(\mathbf{x})$ evaluated after substituting the closed-form position and clock-bias estimate associated with $\alpha$.

\begin{figure}
  \centering
  \includegraphics[width=0.98\linewidth]{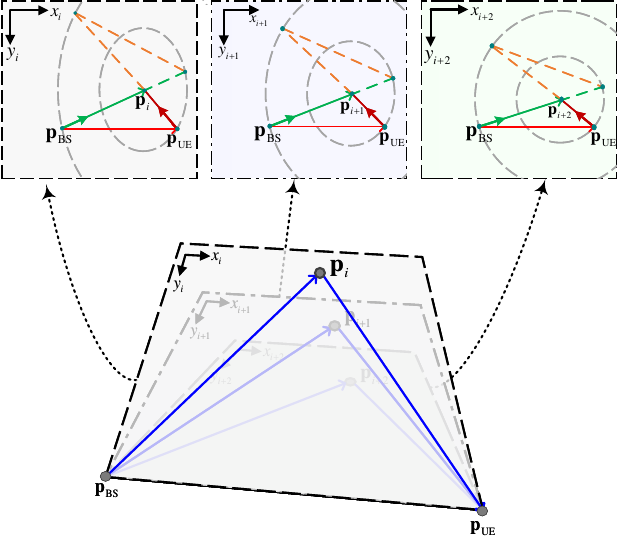} 
  \caption{\small The complete 3-D geometric framework of the considered propagation scenario embedded in the $X-Y-Z$ space, represented by a pencil of planes. Each propagation path defines a local plane spanned by the BS, UE and scatterer $\mathbf{p}
  _i$, and is illustrated using its own local coordinate frame $({x_i},{y_i})$.}
  \label{geometry_3d_framework}
\end{figure}

\subsubsection{\textbf{General Scenario}}
The planar derivation extends directly to 3-D because the BS, UE, and each NLoS-1 scatterer define a local propagation plane. As illustrated by the pencil-of-planes representation in Fig. \ref{geometry_3d_framework} \cite{hartley2003multiple}, the same auxiliary construction yields the global AoD and AoA direction vectors
 \begin{equation}\label{doa unit 3-D}
{{\bf{u}}_i} = {\bf{R}}_{{\rm{BS}}}\left[ \begin{array}{c}
\cos ({\phi _{e,i}})\cos ({\phi _{a,i}})\\
\cos ({\phi _{e,i}})\sin ({\phi _{a,i}})\\
\sin ({\phi _{e,i}})
\end{array} \right],
\end{equation}

\begin{equation}\label{aoa unit 3-D}
{{\bf{v}}_i} = {\bf{R}}_{{\rm{UE}}}\left[ \begin{array}{c}
\cos ({\theta _{e,i}})\cos ({\theta _{a,i}})\\
\cos ({\theta _{e,i}})\sin ({\theta _{a,i}})\\
\sin ({\theta _{e,i}})
\end{array} \right].
\end{equation}

By applying the same geometric relationships in \eqref{solving 2} and \eqref{solving 4} and by adopting the LS criterion in the presence of noise, we have the 3-D counterpart of \eqref{minLS}, that is,
\begin{equation}\label{minLS3d}
\hat{\mathbf{x}} = {\arg\min _{{\mathbf{x}}}}  \textstyle\sum_{i \in \mathcal{I}} {{{\left\| {{{\bf{r}}_i}({\bf{x}})} \right\|}^2}}  \quad  \text{s.t.} \quad  \mathbf{R}_{\mathrm{UE}} \in SO(3),
\end{equation}
where ${\bf{x}} = \{ {{\bf{p}}_{{\rm{UE}}}},{b_{{\rm{UE}}}},{{\bf{R}}_{{\rm{UE}}}}\} $, and the constraint $\mathbf{R}_{\mathrm{UE}}\in SO(3)$ in \eqref{minLS3d} implies $\mathbf{R}_{\mathrm{UE}}^{\mathrm{T}}\mathbf{R}_{\mathrm{UE}} = \mathbf{I}, \det\left(\mathbf{R}_{\mathrm{UE}}\right) = +1$.

As in the planar case, for a fixed $\mathbf R_{\rm UE}$, the position and clock bias admit a closed-form update. Hence, the 3-D problem is handled by profiling out $(\mathbf p_{\rm UE},b_{\rm UE})$ for each orientation candidate and then refining the orientation on $SO(3)$. The orientation initialization and the subsequent local manifold refinement are described below.

\textbf{3-D Position and Clock Bias Estimation:}\label{Position And clock bias Estimation 3-D}
As with the planar scenario, for a given ${\bf{R}_{{\rm{UE}}}}$, a closed-form solution can be obtained for ${{\bf{x}}}_{{\mathbf{R}}_{\rm{UE}}} \doteq {[{\bf{p}}_{{\rm{UE}}}^ \top,{b_{{\rm{UE}}}}]^ \top }$, which is of the same form of \eqref{solving} but has a slightly different $\mathbf{A}$ and $\mathbf{b}$, i.e.,
\begin{equation}\label{AA 3-D}
{\bf{A}} = \left[\cdots, {\mathbf{A}}_i^{\top}, \cdots\right]^{\top} \in {\mathbb{R}}^{(3|\mathcal{I}|)\times 4},\quad i\in\mathcal{I},
\end{equation}
and
\begin{equation}\label{bb 3-D}
{\bf{b}} = \left[\cdots,{{\mathbf{b}}_{i}^ \top},\cdots\right]^ \top  \in {\mathbb{R}}^{(3|\mathcal{I}|)\times 1},\quad i\in\mathcal{I}.
\end{equation}
The LS update has the same normal-equation form as \eqref{solving}, with the planar constraint row removed.
For any fixed orientation $\mathbf R$, let $\hat{\mathbf{x}}_{\mathbf R}\doteq\{\hat{\mathbf p}_{\rm UE}(\mathbf R),\hat b_{\rm UE}(\mathbf R),\mathbf R\}$ denote the full candidate state obtained after this closed-form position and clock-bias update.

\textbf{3-D UE Orientation Estimation:}\label{UE Orientation Estimation 3-D
}
This step estimates the UE orientation matrix by iteratively updating the orientation matrix ${\bf{R }}_{{\rm{UE}}}$. Specifically, at iteration $k$, the orientation matrix is updated as\footnote{Here $\rm{Exp}(\cdot)$ denotes the matrix exponential. For any $\mathbf a\in\mathbb R^3$, $[\mathbf a]_{\times}$ denotes the skew-symmetric matrix satisfying $[\mathbf a]_{\times}\mathbf b=\mathbf a\times\mathbf b$ for any $\mathbf b\in\mathbb R^3$.}
\begin{equation}\label{aue 3-DD}
    {\bf{R}}^{(k + 1)}_{{\rm{UE}}} = \rm{Exp}\left(\left[{\bm{\delta}_{{\rm{UE}}}}\right]_{\times}\right){\bf{R}}^{(k)}_{\rm{UE}} .
\end{equation} 
where ${\bm{\delta}_{{\rm{UE}}}} \in \mathbb{R}^3$ is the left-perturbation increment. The increment vector is obtained as
\begin{equation}\label{delta_UE}
    {{\boldsymbol{\delta }}_{{\rm{UE}}}} =  - \left(\nabla_{\mathbf{R}_{\rm{UE}}}{\mathbf{r}}\right)^{\dagger} \mathbf{r}.
\end{equation}
where $\nabla_{\mathbf{R}_{\rm{UE}}}{\bf{r}}=\left[\cdots,(\nabla_{\mathbf{R}_{\rm{UE}}}{\bf{r}}_i)^\top,\cdots\right]^\top$ for $i\in\mathcal{I}$. Following the left perturbation \cite{barfoot2024state}, the Jacobian matrix of the residual $\mathbf{r}_i(\bf{x})$ with respect to the UE orientation is
\begin{equation}
    {\nabla _{{{\bf{R}}_{{\rm{UE}}}}}} {{\bf{r}}_i} =  - (({\bf{u}}_i^ \top {\bf{t}} - {d_i}){\bf{I}} + {{\bf{u}}_i}{{\bf{t}}^ \top } - {\bf{tu}}_i^ \top ){\left[ {{{\bf{v}}_i}} \right]_ \times } \in \mathbb{R}^{3\times 3}.
\end{equation}
where ${{\bf{v}}_i}$ is the global AoA unit vector in \eqref{aoa unit 3-D}.

\textbf{Orientation Initialization:}
The following procedure provides the initial orientation for the 3-D UE orientation estimation in \eqref{aue 3-DD}. To avoid exhaustive traversal over $SO(3)$ while exploiting the fact that the receiver is usually approximately upright, we use a twist-swing two-stage traversal. Here, the twist denotes an in-plane yaw rotation, while the swing denotes a bounded out-of-plane tilt around a horizontal axis. Using the planar yaw block in \eqref{eq:planar_yaw_rotation}, the twist rotation and the swing rotation are defined as
\begin{equation}\label{eq:twist_swing_rotation_init}
\mathbf{R}_{\rm tw}(\psi)=\rm{blkdiag}\left(\tilde{\mathbf{R}}(\psi),1\right),\,
\mathbf{R}_{\rm sw}(\chi,\beta)=\operatorname{Exp}\left(\beta[\mathbf a(\chi)]_{\times}\right),
\end{equation}
where $\psi$ is the twist angle, $\mathbf a(\chi)=\left[-\sin\chi,\cos\chi,0\right]^\top$ specifies the horizontal swing axis, and $\beta$ is the swing magnitude. 

An upright UE implies that the local vertical axis of the receiver is approximately aligned with the global vertical axis, i.e., $\mathbf R_{\rm UE}\mathbf e_3\approx\mathbf e_3$ with $\mathbf e_3=[0,0,1]^\top$. Thus, the dominant attitude uncertainty is the twist angle around the vertical axis, while the swing angle $\beta$ only accounts for a bounded tilt away from the upright direction.
Compared with an Euler-angle parameterization, this decomposition avoids rotation-order ambiguity and makes the upright prior explicit, because $\beta$ directly bounds the tilt of the receiver vertical axis instead of indirectly constraining coupled pitch and roll angles.
We define the twist grid $\mathcal A_\psi=\{-\pi,-\pi+\Delta_\psi,\ldots,\pi\}$, the swing-axis grid $\mathcal A_\chi=\{0,\Delta_\chi,\ldots,2\pi\}$, and a bounded swing-magnitude grid $\mathcal A_\beta=\{0,\Delta_\beta,\ldots,\beta_{\max}\}$.

In the first stage, the twist angle is selected by
\begin{equation}
\hat{\psi}=\arg\min_{\psi\in\mathcal A_\psi}
J\!\left(\hat{\mathbf{x}}_{\mathbf R_{\rm tw}(\psi)}\right),
\label{eq:twist_stage_init}
\end{equation}
where $\mathbf p_{\rm UE}$ and $b_{\rm UE}$ are updated in closed form for each twist candidate using \eqref{AA 3-D}--\eqref{bb 3-D}. In the second stage, a bounded swing traversal is performed around the selected twist seed:
\begin{equation}
(\hat\chi,\hat\beta)=\arg\min_{\chi\in\mathcal A_\chi, \beta\in\mathcal A_\beta}
J\!\left(\hat{\mathbf{x}}_{\mathbf R_{\rm sw}(\chi,\beta)\mathbf R_{\rm tw}(\hat\psi)}\right).
\label{eq:twist_swing_init}
\end{equation}
Starting from $\mathbf R_{\rm UE}^{(0)}=\mathbf R_{\rm sw}(\hat\chi,\hat\beta)\mathbf R_{\rm tw}(\hat\psi)$, the subsequent Gauss-Newton refinement updates the full rotation on $SO(3)$.

\textbf{Local Convergence:} The first block is a strictly convex least-squares problem, thus, its update does not increase the objective. For the second block, in a neighborhood of the solution, the convergence rates of the Gauss-Newton method on $SO(3)$ is superlinear \cite{absil2009optimization}.

\subsubsection{\textbf{Feasibility Check}}\label{Feasibility Check}
During the coarse search, feasibility checks reject geometrically inconsistent candidates without requiring prior LoS knowledge. The final LoS decision is deferred to the model-selection refinement in Section \ref{Iterative Estimation and LoS Detection}. Each path is tested against two admissible inlier interpretations:
\begin{itemize}
    \item A collinearity-consistent LoS: when the norm of both  $\Delta_1= \hat{\bar{\mathbf{t}}} \times \hat{\mathbf{v}}_i$ and $\Delta_2 = \hat{\bar{\mathbf{t}}} \times {\mathbf{u}}_i$ are below threshold $\varepsilon_c$ in absolute value, where $\hat{\bar{\bf{t}}}$ is the unit vector of $\hat{\bf{t}}=({\hat {\bf{p}}_{{\rm{UE}}}} - {{\bf{p}}_{{\rm{BS}}}})$, the transmit-receive pair is confirmed to be collinear with $\hat{\bf{t}}$.
    \item A single-bounce-consistent NLoS-1: when $\tfrac{\Delta_1^{\top} \Delta_2}{\|\Delta_1\|\| \Delta_2\|}$ is above a positive threshold $\varepsilon_p$, the vectors $\mathbf{u}_i$ and $\mathbf{v}_i$ reside on the same side of the vector $\hat{\bf{t}}$ within a common plane. Upon confirmation of this, if the inner products further satisfy $\hat{\bar{\mathbf{t}}}^{\top} \hat{\mathbf{v}}_i<\hat{\bar{\mathbf{t}}}^{\top} \mathbf{u}_i$, the path is considered satisfying the NLoS-1 path condition.
\end{itemize}
A candidate UE state is retained only if every selected inlier path satisfies at least one of these conditions.

\subsection{Extension to Scenarios in the presence of Outliers}\label{Extension to Amplitude-Independent Method with Outliers}
To handle NLoS-n outliers, we embed the coarse solver in a consensus-based exhaustive minimal-set search. The robust search is related to earlier snapshot SLAM ideas \cite{10818978}, but its hypothesis generation and feasibility handling are tailored to the proposed angle-delay formulation. Algorithm \ref{algorithm 1} uses an MSAC-style truncated cost \cite{TORR2000138}; the minimum-cost subset defines the coarse inlier set and the initial UE-state estimate for refinement.

\begin{algorithm}[t]
\caption{Robust coarse solver}
\label{alg:robust_coarse_init}
\begin{algorithmic}[1]
\Require Observation set $\mathcal Z$, index set $\mathcal I$, minimal subset size $N_{\rm min}$, threshold $T_{\varepsilon}$, dimension flag $D\in\{2\mathrm{D},3\mathrm{D}\}$, orientation grids $\mathcal A$, $\mathcal A_\psi$, $\mathcal A_\chi$, and $\mathcal A_\beta$
\Ensure Coarse inlier set $\widehat{\mathcal I}$ and initial UE-state estimate $\hat{\mathbf x}_0$

\State Generate the set of all candidate minimal subsets

$\mathfrak J=\left\{\mathcal J\subseteq \mathcal I:\ |\mathcal J|=N_{\rm min}\right\}$

\State Initialize $C^\star \gets +\infty$, $\widehat{\mathcal I}\gets \emptyset$, $\hat{\mathbf x}_0\gets \varnothing$

\ForAll{$\mathcal J\in \mathfrak J$}
    \If{$D=2\mathrm{D}$}
        \State Traverse $\alpha_{\rm UE}\in\mathcal A$ and update $\mathbf p_{\rm UE},b_{\rm UE}$ using \eqref{solving}
        \State Recover the orientation using \eqref{aue}
    \Else
        \State Traverse the twist grid $\mathcal A_\psi$ and update $\mathbf p_{\rm UE},b_{\rm UE}$ using \eqref{AA 3-D}--\eqref{bb 3-D}
        \State Traverse the swing grids $\mathcal A_\chi$ and $\mathcal A_\beta$ using \eqref{eq:twist_swing_init}
        \State Refine the full orientation using \eqref{aue 3-DD}
    \EndIf
    \State \textbf{continue} if $\mathcal J$ fails the feasibility check
    \State Form the coarse UE-state hypothesis $\hat{\mathbf x}_{\mathcal J}$

    \State Initialize the consensus set $\mathcal I_{\mathcal J}\gets \emptyset$
    \ForAll{path $k\in \mathcal I$}
        \State Compute the residual $\mathbf{r}_k(\hat{\mathbf x}_{\mathcal J})$
        \If{$\|\mathbf{r}_k(\hat{\mathbf x}_{\mathcal J})\|<T_{\varepsilon}$}
            \State $\mathcal I_{\mathcal J}\gets \mathcal I_{\mathcal J}\cup\{k\}$
        \EndIf
    \EndFor
    
    \State Compute the truncated cost
        $C_{\mathcal J}=\sum_{k\in \mathcal I}\min\left(\|\mathbf{r}_k(\hat{\mathbf x}_{\mathcal J})\|^2,\,T_{\varepsilon}^2\right)$

    \If{$C_{\mathcal J}<C^\star$}
        \State $C^\star\gets C_{\mathcal J}$, $\widehat{\mathcal I}\gets \mathcal I_{\mathcal J}$, $\hat{\mathbf x}_0\gets \hat{\mathbf x}_{\mathcal J}$
    \EndIf
\EndFor

\State \Return $\widehat{\mathcal I},\hat{\mathbf x}_0$
\end{algorithmic}
\label{algorithm 1}
\end{algorithm}

\subsection{Iterative Estimation and LoS Detection}\label{Iterative Estimation and LoS Detection}
The coarse stage provides an inlier set and an initial UE-state estimate, while final estimation fits the physical measurement model to $\mathcal{Z}_{\rm{inliers}}$. Since LoS and NLoS-1 paths have different parameterizations, refinement and LoS identification are handled jointly by model selection \cite{8498082}.
With $M$ inliers, we consider $M+1$ candidate models: one all-NLoS-1 model and $M$ mixed models in which path $m$ is interpreted as LoS and the remaining paths as NLoS-1. For each candidate model, $\mathbf{h}_k^{(m)}$ represents the chosen parameterization of the $k$th observation, and the parameter set is
\begin{equation}
{\bf{\Theta}}^{(m)} = 
\begin{cases} 
\{\mathbf{x}\}\cup\{{\bf{p}}_{k}|k\in\mathcal{I}_{\rm{inliers}}\backslash m\}, & \rm{for} \; \rm{LoS} \; \rm{case}\\
\{\mathbf{x}\}\cup\{{\bf{p}}_{k}|k\in\mathcal{I}_{\rm{inliers}}\},& \rm{otherwise}
\end{cases}.
\end{equation}
For model $m$, define the stacked residual and measurement Jacobian as $\mathbf e^{(m)}(\boldsymbol{\Theta}^{(m)})
=\mathbf z-\mathbf h^{(m)}(\boldsymbol{\Theta}^{(m)})$ and $\mathbf J^{(m)}(\boldsymbol{\Theta}^{(m)})
=\nabla_{\boldsymbol{\Theta}^{(m)}}\mathbf h^{(m)}(\boldsymbol{\Theta}^{(m)})$.
Because calibrated path-wise covariance matrices are not used as algorithmic inputs, the refinement is implemented as an iteratively reweighted least-squares (IRLS) procedure based on Jacobian-row equilibration \cite{Gambarini2026,Dixon2011}, rather than as Gauss-Newton minimization of a fixed covariance-weighted likelihood. At iteration $\ell$, the residual and Jacobian are evaluated at the current linearization point,
\begin{equation}
\mathbf e_{\ell}^{(m)}
=\mathbf e^{(m)}(\boldsymbol{\Theta}_{\ell}^{(m)}),\quad
\mathbf J_{\ell}^{(m)}
=\mathbf J^{(m)}(\boldsymbol{\Theta}_{\ell}^{(m)}).
\end{equation}
To mitigate local sensitivity imbalance among residual components while avoiding excessive amplification of weakly sensitive or nearly singular residual equations, the Ruiz row-equilibration weight \cite{ruiz2001scaling} of the $j$th residual component is computed as $\mathbf W_{\ell}^{(m)}
=\operatorname{diag}\!\left([\ldots,\|\mathbf J_{k,\ell}^{(m)}\|^{-1},\ldots]\right)$, where $\mathbf J_{k,\ell}^{(m)}$ is the $k$th row of $\mathbf J_{\ell}^{(m)}$. The weights are recomputed at the beginning of each IRLS iteration and then frozen within the local least-squares step. The local increment is obtained from
$\bm{\delta}_{\ell}^{(m)}
=
\left((\mathbf W_{\ell}^{(m)})^{1/2}\mathbf J_{\ell}^{(m)}\right)^{\dagger}
\left(\mathbf W_{\ell}^{(m)})^{1/2}\mathbf e_{\ell}^{(m)}\right)$. 
The increment $\bm{\delta}_{\ell}^{(m)}$ is applied by ordinary addition for Euclidean variables and by the exponential-map update for $\mathbf R_{\rm UE}\in SO(3)$ \cite{absil2009optimization}. After convergence at iteration $\ell^\star$, the model-$m$ estimate is $\hat{\boldsymbol{\Theta}}^{(m)}=\boldsymbol{\Theta}_{\ell^\star}^{(m)}$, and the final frozen weighted residual score is
\begin{equation}
\widetilde E_m(\hat{\boldsymbol{\Theta}}^{(m)})
=
\left\|
(\mathbf W_{\ell^\star}^{(m)})^{\tfrac{1}{2}}
\mathbf e^{(m)}
\right\|^2 .
\end{equation}
The associated Jacobian matrices are given in Appendix \ref{Appendix A}.

Since $\widetilde{E}_m$ is the terminal score of a row-equilibrated IRLS procedure rather than the exact negative Gaussian log-likelihood, it is used only as a quasi-deviance for model comparison. All $M+1$ candidate models are refined first, and the LoS hypothesis is then selected by the quasi-Akaike information criterion (QAIC) \cite{8498082}:
\begin{equation}\label{AIC 2}
    \mathrm{QAIC}_{m} = 2d_{\Theta}^{(m)} + \tfrac{n}{\hat{c}}\ln\left( \tfrac{\widetilde{E}_m}{n} \right),
\end{equation}
where $n=k_d|\mathcal{I}_{\rm{inliers}}|$, with $k_d=3$ (or 5) for 2-D (or 3-D) measurements, $d_{\Theta}^{(m)}$ is the model degrees of freedom, and $\hat c$ is the residual scale. The model with the smallest \eqref{AIC 2} determines the LoS state.

\subsubsection{Initialization} The refinement variables are initialized from the coarse-stage output. Specifically, the UE state is set to the coarse estimate returned by Algorithm \ref{algorithm 1}, while the scattering points, which are not explicit variables in the coarse formulation, are recovered analytically from this coarse UE-state estimate.
 
\textbf{UE state} ${\bf{x}}$:
The coarse UE-state estimate is used directly, i.e., $\mathbf{x}^{(0)}=\hat{\mathbf{x}}_0$, so no external UE-state initialization is required.

\textbf{Scattering Point} ${\mathbf{p}}_i$:
Given the coarse UE position $\hat{\bf{p}}_{\rm{UE}}$ and orientation $\hat{\mathbf{R}}_{\rm{UE}}$, each retained NLoS-1 scattering point is initialized as

\begin{equation}\label{Pi}
    {{\bf{p}}_i} = 
\mathbf{p}_{\rm{BS}}+\frac{d_{l,i}\mathbf{u}_i+d_{r,i}\hat{\mathbf{v}}_i+\hat{\mathbf{t}}}{2},
\end{equation}
where $[{{d_{l,i}}}, {{d_{r,i}}}]^{\top} = {\left( {\left[ {{{\bf{u}}_i},\; - {{\hat {\bf{v}}}_i}} \right]} \right)^\dag }\hat {\bf{t}}$.

\subsubsection{Computational Complexity}\label{computational complexity}

For the closed-form-like algorithm, given $N$ observed paths, the consensus-based exhaustive minimal-set search tests all $N_\mathrm{min}$-subsets of the $N$ paths, which yields $L= \tfrac{{|{\cal Z}|!}}{{(|{\cal Z}| - {N_{{\rm{min}}}})!{N_{{\rm{min}}}}!}}.$ minimal-set hypotheses. For each hypothesis, the 2-D and 3-D estimations adopt a slightly different strategy. On the one hand, a candidate 2-D solution is obtained via a discrete search over ${{\mathcal{|A|}}}$ UE headings. For each heading, the closed-form update of 
\eqref{solving} is obtained by solving a pseudoinverse whose dominant complexity is of order $\mathcal{O}\left(N_{\rm{min}}^2\right)$. Overall, the 2-D estimation has complexity
$\mathcal{O}(L{{\cal|A|}}N_{\rm{min}}^2)$. On the other hand, a 3-D candidate is obtained from the minimal subset via $K_{\mathrm{cl}}$ alternating updates. For each alternating update, the closed-form update
\eqref{solving} is of order $\mathcal{O}\left(N_{\rm{min}}^2\right)$, while evaluating the perturbation vector of \eqref{delta_UE} has a dominant complexity of order $\mathcal{O}\left(N_{\rm{min}}^2\right)$. Consequently, the overall 3-D estimation has complexity $\mathcal{O}(LK_{\rm{cl}}N_{\rm{min}}^2)$.

Finally, for the iterative estimation, the complexity is dominated by the computation of the pseudoinverse.
Assume $M$ inliers are selected with the closed-form-like algorithm, which yields $M+1$ candidate models,
the pseudoinverse computed in each iteration scales on the order of $\mathcal{O}(M^3)$. With $K_\mathrm{ls}$ iterations, the refinement complexity is $\mathcal{O}({K_\mathrm{ls}}(M+1)M^3)$.

\section{Formulation-Specific Local-Rank and Minimal-Sample Implications}\label{Formulation-Specific Local-Rank and Minimal-Sample Implications}
The radio SLAM problem is only identifiable if a sufficient number of inliers exist. Herein, we first analyze the minimum number of inliers required to ensure a unique solution for ${\bf{x}}$ in the 3-D scenario, based on which the required number for the planar scenario is then derived.

\subsection{3-D Scenario}
To analyze the minimum number of inliers required by the proposed formulation, we examine the rank and nullspace of the per-path Jacobian. The rank indicates how many independent first-order constraint directions an inlier can provide, whereas the nullspace identifies the UE-state perturbations that remain locally unobservable from that path.
The Jacobian matrix of the residual error $\mathbf{r}_i(\mathbf{x})$ is
\begin{equation}\label{Jx}
\nabla_{\mathbf{x}}\mathbf{r}_i
= 
\Bigl[
\underbrace{\mathbf{M}_i}_{\nabla_{{\mathbf{p}}_{\rm{UE}}}\mathbf{r}_i}
\quad
\underbrace{c(\mathbf{v}_i - \mathbf{u}_i)}_{\nabla_{{b}_{\rm{UE}}}\mathbf{r}_i}
\quad
\underbrace{-  \mathbf{G}_i {\left[ {{{\bf{v}}_i}} \right]_ \times }}_{\nabla_{{\mathbf{R}}_{\rm{UE}}}\mathbf{r}_i}
\Bigr]\in\mathbb{R}^{3\times 7},
\end{equation}
where $\mathbf{M}_i$ is given by \eqref{M}, and
$\mathbf{G}_i = \big(\mathbf{u}_i^\top \mathbf{t} - d_i\big)\mathbf{I} + \mathbf{u}_i \mathbf{t}^\top - \mathbf{t} \mathbf{u}_i^\top \in\mathbb{R}^{3\times 3}$.

\begin{theorem}[Single-Measurement Jacobian Rank Bounds]\label{thm:Ji-rank}
Let $\mathbf{n}_i = \mathbf{u}_i+\mathbf{v}_i$. For the Jacobian $\nabla_{\mathbf{x}}\mathbf{r}_i$, the following statements hold:
\begin{enumerate}
\item For a NLoS-1 path, i.e., $\mathbf{u}_i\neq -\mathbf{v}_i$, then
\begin{equation}\label{eq:Ei-rank-null}
\mathrm{rank}\big(\nabla_{{\mathbf{p}}_{\rm{UE}}}\mathbf{r}_i\big)=2,\quad \mathcal{N}\big(\nabla_{{\mathbf{p}}_{\rm{UE}}}\mathbf{r}_i\big)=\mathrm{span}\{\mathbf{n}_i\}.
\end{equation}
For a LoS path, i.e., $\mathbf{u}_i=-\mathbf{v}_i$, ${\rm{rank}}\big(\nabla_{{\mathbf{p}}_{\rm{UE}}}\mathbf{r}_i\big)=1$.
\item The sub-Jacobian $\left[\nabla_{{\mathbf{p}}_{\rm{UE}}}\mathbf{r}_i\;\nabla_{{b}_{\rm{UE}}}\mathbf{r}_i\right]$ satisfies
\begin{equation}\label{eq:Jtb-rank}
\operatorname{rank}\!\big([\nabla_{{\mathbf{p}}_{\rm{UE}}}\mathbf{r}_i\;\nabla_{{b}_{\rm{UE}}}\mathbf{r}_i]\big)
=\operatorname{rank}(\nabla_{{\mathbf{p}}_{\rm{UE}}}\mathbf{r}_i).
\end{equation}
\item 
The rotation sub-Jacobian $\nabla_{{\mathbf{R}}_{\rm{UE}}}\mathbf{r}_i$ contributes a component along $\mathbf{n}_i$ if and only if
\begin{equation}\label{eq:out-of-plane-cond}
\mathbf{n}_i^{\top}\left(\nabla_{{\mathbf{R}}_{\rm{UE}}}\mathbf{r}_i\right)\neq 0
\quad \Longleftrightarrow\quad
\mathbf{v}_i\times (\mathbf{G}_i^{\top}\mathbf{n}_i)\neq \mathbf{0}.
\end{equation}
\end{enumerate}
\end{theorem}
\noindent \textit{Geometric Interpretation.} Theorem~1 shows that a single path contributes only a low-dimensional path-induced constraint subspace. In particular, an NLoS-1 path leaves a sliding direction $\mathrm{span}\{\mathbf n_i\}$, whereas a LoS path is more degenerate and contributes only one independent direction. The clock-bias column does not create an additional single-path constraint direction.
\begin{proof}
First, consider the basis $\{\mathbf{n}_i,\mathbf{g}_i,\mathbf{q}_i\}$ with $\mathbf{g}_i=\mathbf{u}_i-\mathbf{v}_i$, and $\mathbf{q}_i=\mathbf{u}_i\times \mathbf{v}_i$.
A direct substitution into~\eqref{M} yields
$\mathbf{M}_i \mathbf{n}_i=\mathbf{0}$, $\mathbf{M}_i \mathbf{g}_i=-2\mathbf{g}_i$, and $\mathbf{M}_i \mathbf{q}_i=-(1+\mathbf{u}_i^{\top}\mathbf{v}_i)\mathbf{q}_i$.
Hence, when $\mathbf{u}_i\neq -\mathbf{v}_i$ we have one zero eigenvalue and two nonzero eigenvalues, implying
$\operatorname{rank}(\mathbf{M}_i)=2$ and $\mathcal{N}(\mathbf{M}_i)=\mathrm{span}\{\mathbf{n}_i\}$.
If $\mathbf{u}_i=-\mathbf{v}_i$, then $\mathbf{q}_i=\mathbf{0}$ and only $\mathbf{g}_i$ remains as a nonzero-eigen direction, giving rank one.
Second, from~\eqref{M} and $\|\mathbf{u}_i\|=\|\mathbf{v}_i\|=1$ we obtain the identity
\begin{equation}\label{eq:Ei-ui}
\mathbf{M}_i\mathbf{u}_i = \mathbf{v}_i-\mathbf{u}_i,
\end{equation}
which implies $-(\mathbf{v}_i-\mathbf{u}_i)=-\mathbf{M}_i \mathbf{u}_i$ and hence the column $\nabla_{{b}_{\rm{UE}}}\mathbf{r}_i$ lies in the column space of $\nabla_{{\mathbf{p}}_{\rm{UE}}}\mathbf{r}_i$.
Therefore, concatenating $\nabla_{{b}_{\rm{UE}}}\mathbf{r}_i$ cannot increase the rank beyond $\operatorname{rank}(\nabla_{{\mathbf{p}}_{\rm{UE}}}\mathbf{r}_i)$, proving~\eqref{eq:Jtb-rank}.
Finally, for the rotation sub-Jacobian, the left-perturbation form satisfies
\begin{equation}
\nabla_{{\mathbf{R}}_{\rm{UE}}}\mathbf{r}_i
= -\mathbf{G}_i[\mathbf{v}_i]_{\times} .
\end{equation}
Applying the scalar triple-product identity gives
\begin{equation}
\mathbf{n}_i^{\top}\mathbf{G}_i(\bm{\delta}\times \mathbf{v}_i)
=\bm{\delta}^{\top}\big(\mathbf{v}_i\times(\mathbf{G}_i^{\top}\mathbf{n}_i)\big),
\end{equation}
which proves the equivalence~\eqref{eq:out-of-plane-cond}.
\end{proof}

\begin{corollary}[No Rank Contribution at a Perfect Fit]\label{cor:JR}
Under the assumptions of Theorem~\ref{thm:Ji-rank}, if the model is perfectly satisfied
for measurement $i$, i.e.,
\begin{equation}\label{eq:perfect-fit}
\mathbf{r}_i(\mathbf{x}^\star)=\mathbf{0},
\end{equation}
then the rotation sub-Jacobian has nullspace $\mathrm{span}\{\mathbf{v}_i^\star\}$ and does not contribute any component outside $\mathbf{n}_i^{\perp}$:
\begin{equation}\label{eq:no-out-of-plane}
\mathbf{n}_i^{\top}\left(\nabla_{{\mathbf{R}}_{\rm{UE}}}\big|_{\mathbf{x}^\star}\mathbf{r}_i\right)=\mathbf{0}
\;\Longrightarrow\;
\mathrm{Col}\big(\nabla_{{\mathbf{R}}_{\rm{UE}}}\big|_{\mathbf{x}^\star}\mathbf{r}_i\big)\subseteq \mathbf{n}_i^{\perp}.
\end{equation}
Consequently, the single-measurement Jacobian evaluated at a perfect fit has effective rank of $\mathrm{rank}\big(\nabla_{{\mathbf{p}}_{\rm{UE}}}\mathbf{r}_i\big)$,
in the sense that all first-order variations of $\mathbf{r}_i$ lie in the same 2-D subspace $\mathbf{n}_i^{\perp}$.
\end{corollary}
\noindent \textit{Geometric Interpretation.} Near a consistent solution, the rotational perturbation remains confined to the same path-induced subspace $\mathbf n_i^\perp$. Hence, even the rotational contribution of a single path does not generally introduce a new independent first-order direction.
\begin{proof}[Proof]
At a perfect fit~\eqref{eq:perfect-fit}, for a LoS path, i.e., $\mathbf{u}^\top_i\mathbf{t}^{\star} = d_i^\star$, we have $\nabla_{{\mathbf{R}}_{\rm{UE}}}\mathbf{r}_i= \mathbf{0}$. On the other hand, for an NLoS-1 path, we have $\mathbf{M}_i^\star\mathbf{t}^\star=d_i^\star(\mathbf{v}_i^\star-\mathbf{u}_i)$.
This equality implies that $\mathbf{t}^\star$ lies in the plane spanned by $\{\mathbf{u}_i,\mathbf{v}_i^\star\}$ (i.e., it has no component
along $\mathbf{u}_i\times \mathbf{v}_i^\star$). Substituting this structure into the left-perturbation Jacobian gives $\bm{\delta}\parallel \mathbf{v}_i^\star$, which implies $\mathcal{N}\big(\nabla_{{\mathbf{R}}_{\rm{UE}}}\big|_{\mathbf{x}^\star}\mathbf{r}_i\big)=\mathrm{span}\{\mathbf{v}_i^\star\}$. Substituting the perfect-fit structure into $\mathbf{G}_i^{\top}\mathbf{n}_i$ gives $\mathbf{G}_i^{\top}\mathbf{n}_i\parallel \mathbf{v}_i^\star$, and hence $\mathbf{v}_i^\star\times(\mathbf{G}_i^{\top}\mathbf{n}_i)=\mathbf{0}$.
By~\eqref{eq:out-of-plane-cond}, this yields $\mathbf{n}_i^{\top}\left(\nabla_{{\mathbf{R}}_{\rm{UE}}}\big|_{\mathbf{x}^\star}\mathbf{r}_i\right)=\mathbf{0}$ and~\eqref{eq:no-out-of-plane}.
\end{proof}

\begin{corollary}[Minimal Sample Size $N_{\rm{min}}$]\label{cor:min-sample}
For the proposed unified coarse-stage formulation, an NLoS-1 measurement provides at most two independent constraint directions on $\mathbf{x}$, whereas a LoS path provides only one (Theorem~\ref{thm:Ji-rank}). Since the total number of degrees of freedom is $7$, for the proposed coarse solver, the resulting local-rank heuristic suggests
\begin{equation}\label{eq:min-sample}
\left\{\begin{array}{cc}
2N_{\rm{min}} \ge 7 & \mathrm{for\,\,NLoS} \\
2(N_{\rm{min}}-1) + 1 \ge 7 & \mathrm{for\,\,LoS}
\end{array}\right.
\; \Rightarrow \; N_{\rm{min}}\ge 4.
\end{equation}
\end{corollary}

This statement concerns the proposed residual/Jacobian analysis and should not be interpreted as a fundamental minimum for the broader problem class.

\begin{theorem}[Stacked Nullspace Criterion for Local Identifiability]\label{thm:full-rank}
Let $\nabla_{\mathbf{x}}{\bf{r}}\in\mathbb{R}^{3N\times 7}$ be the stacked Jacobian obtained by vertically concatenating~\eqref{Jx} over $i=1,\ldots,N$. Consider a consistent solution $\mathbf{x}^{\star}$ and evaluate all quantities below at $\mathbf{x}^{\star}$, with the superscript omitted for compactness. Define
\begin{equation}\label{eq:BiPi}
\mathbf{B}_i =
\left[
c(\mathbf{v}_i-\mathbf{u}_i)
\quad
-\mathbf{G}_i[\mathbf{v}_i]_{\times}
\right]\in\mathbb{R}^{3\times4},
\quad
\mathbf{P}_i=\mathbf{M}_i^{\dagger}\mathbf{M}_i,
\end{equation}
where $\mathbf{P}_i$ is the orthogonal projector onto $\operatorname{Range}(\mathbf{M}_i)$. If the reduced stacked nullspace matrix
\begin{equation}\label{eq:H-nullspace}
\mathbf{H}=
\begin{bmatrix}
\mathbf{P}_1 & \mathbf{M}_1^{\dagger}\mathbf{B}_1\\
\vdots & \vdots\\
\mathbf{P}_N & \mathbf{M}_N^{\dagger}\mathbf{B}_N
\end{bmatrix}\in\mathbb{R}^{3N\times 7}
\end{equation}
has full column rank, then $\nabla_{\mathbf{x}}{\bf{r}}$ has full column rank and the UE state $\mathbf{x}$ is locally identifiable.
\end{theorem}
\noindent \textit{Geometric Interpretation.} The matrix $\mathbf H$ makes explicit the possible coupling between translation, clock bias, and rotation perturbations. For an NLoS-1 path, $\mathbf P_i$ projects away the sliding direction $\mathrm{span}\{\mathbf n_i\}$; for a LoS path, it projects onto the one-dimensional LoS constraint direction. Full column rank of $\mathbf H$ means that no nonzero joint perturbation can remain hidden in all path-induced nullspaces simultaneously.
\begin{proof}[Proof]
Let a first-order perturbation be written as
$\delta\mathbf{x}=[\delta\mathbf{p}^{\top},\delta b,\delta\boldsymbol{\omega}^{\top}]^{\top}$ and set $\delta\mathbf{q}=[\delta b,\delta\boldsymbol{\omega}^{\top}]^{\top}$.
For the $i$th measurement,
\begin{equation}
\nabla_{\mathbf{x}}\mathbf{r}_i\,\delta\mathbf{x}
=
\mathbf{M}_i\delta\mathbf{p}+\mathbf{B}_i\delta\mathbf{q}.
\end{equation}
Assume $\delta\mathbf{x}\in\mathcal{N}(\nabla_{\mathbf{x}}{\bf r})$. Then for every $i$,
\begin{equation}\label{eq:null-each-path}
\mathbf{M}_i\delta\mathbf{p}+\mathbf{B}_i\delta\mathbf{q}=\mathbf{0}.
\end{equation}
By Theorem~\ref{thm:Ji-rank}, the clock-bias column lies in $\operatorname{Range}(\mathbf{M}_i)$, and by Corollary~\ref{cor:JR}, at a consistent solution the rotation block also lies in the same path-induced constraint subspace. Hence $\mathbf{B}_i\delta\mathbf{q}\in\operatorname{Range}(\mathbf{M}_i)$, so applying $\mathbf{M}_i^{\dagger}$ to~\eqref{eq:null-each-path} gives
\begin{equation}
\mathbf{P}_i\delta\mathbf{p}+\mathbf{M}_i^{\dagger}\mathbf{B}_i\delta\mathbf{q}=\mathbf{0}.
\end{equation}
Stacking these equations over all measurements yields $\mathbf H\delta\mathbf{x}=\mathbf 0$. If $\mathbf H$ has full column rank, then $\delta\mathbf{x}=\mathbf 0$. Therefore $\mathcal{N}(\nabla_{\mathbf{x}}{\bf r})=\{\mathbf 0\}$, which proves that the stacked Jacobian has full column rank. Local identifiability follows from the standard local-rank condition.
\end{proof}
\noindent \textit{Remark.} The formal identifiability condition is the full-column-rank test on the stacked matrix $\mathbf H$ in~\eqref{eq:H-nullspace}, which explicitly accounts for coupled translation-clock-rotation null directions. Informally, this full-rank condition requires sufficient geometric diversity in $\{\mathbf v_i\}$ and $\{\mathbf n_i=\mathbf u_i+\mathbf v_i\}$. In particular, diversity of $\{\mathbf v_i\}$ suppresses common rotational ambiguity, diversity of $\{\mathbf n_i\}$ suppresses common translational sliding directions, and at least one non-degenerate angle pair with $\mathbf v_i\neq\mathbf u_i$ helps make the clock-bias direction observable.

\subsection{Planar Scenario} Assume a planar setting in which all quantities lie in the $X-Y$ plane and rotation is restricted to yaw only: $\mathbf{t}=\left[\tilde{\mathbf{t}}^\top,0\right]^\top$, $\mathbf{u}=\left[\tilde{\mathbf{u}}^\top,0\right]^\top$, $\mathbf{v}=\left[\tilde{\mathbf{v}}^\top,0\right]^\top$, $\mathbf{R}_{\rm{UE}} = \rm{blkdiag}\left(\tilde{\mathbf{R}}(\alpha_{\rm{UE}}),1\right)$, then substituting into \eqref{solving 4} shows that, for every measurement $i$, the 3-D residual satisfies $\mathbf{r}_i=[\tilde{\mathbf{r}}_i^\top,0]^\top$, where the induced residual in plane is
\begin{equation}  \tilde{\mathbf{r}}_i(\tilde{\mathbf{t}},b,\alpha_{\rm{UE}})=\tilde{\mathbf{M}}_i\tilde{\mathbf{t}} - d_i(\tilde{\mathbf{v}}_i-\tilde{\mathbf{u}}_i)\in\mathbb{R}^2,
\end{equation}
where $\tilde{\mathbf{M}}_i=\tilde{\mathbf{v}}_i \tilde{\mathbf{u}}_i^\top+\tilde{\mathbf{u}}_i \tilde{\mathbf{v}}_i^\top-(\tilde{\mathbf{u}}_i^\top \tilde{\mathbf{v}}_i+1)\mathbf{I}_2$.
\begin{theorem}[Rank and Nullspace Reduction in Planar Scenario]
Under the planar assumptions, the following statements hold:
\begin{enumerate}
\item the planar matrix $\tilde{\mathbf{M}}_i$ satisfies
\begin{equation}
    \operatorname{rank}(\tilde{\mathbf{M}}_i)=1,\quad 
\mathcal N(\tilde{\mathbf{M}}_i)=\mathrm{span}\{\tilde{\mathbf{n}}_i\},
\end{equation}
and $\mathrm{Range}(\tilde{\mathbf{M}}_i)=\tilde{\mathbf{n}}_i^\perp\subset\mathbb{R}^2$ is a 1-D line orthogonal to $\tilde{\mathbf{n}}_i$, where $\tilde{\mathbf{n}}_i=\tilde{\mathbf{u}}_i+\tilde{\mathbf{v}}_i$.
\item Let the unknown vector be $\tilde{\mathbf{x}}=[\tilde{\mathbf{p}}_{UE}^\top,\ b_{\rm{UE}},\ \alpha_{\rm{UE}}]^\top\in\mathbb{R}^4$, near a consistent solution, the per-measurement Jacobian ${\nabla _{\tilde{\mathbf{x}}}}{\tilde{\mathbf{r}}_i}$ is effectively rank one.
\end{enumerate}
\end{theorem}

\noindent \textit{Geometric Interpretation.} In the planar reduction, each path contributes only one constraint line, while the solution remains locally free along $\tilde{\mathbf n}_i$. Thus, the planar case is the reduced counterpart of the 3-D path-induced constraint-subspace interpretation.

\begin{proof}[Proof] First, under planar reduction, the top-left block $\tilde{\mathbf{M}}_i$ is precisely the restriction of $\mathbf{M}_i$ to the $X-Y$ subspace. Since the planar subspace is 2-D, the corresponding range space is the intersection $\tilde{\mathbf{n}}_i^\perp\cap\mathbb{R}^2$, which is a 1-D subspace (a line). Hence $\operatorname{rank}(\tilde{\mathbf{M}}_i)=1$ and $\mathcal N(\tilde{\mathbf{M}}_i)=\mathrm{span}\{\tilde{\mathbf{n}}_i\}$. Second, due to Theorem~\ref{thm:Ji-rank}, the column associated with 
$b_{\rm{UE}}$, always lies in $\mathrm{Range}(\tilde{\mathbf{M}}_i)$. Further, since $\mathrm{Range}(\tilde{\mathbf{M}}_i)$ is 1-D in 2-D, the rotation contribution cannot provide a new independent direction beyond that line near consistency (Theorem~\ref{thm:Ji-rank}); hence the per-measurement linear information becomes effectively one-dimensional.
\end{proof}

\begin{corollary}[Minimal sample size heuristic in plane]
In the planar specialization, the proposed coarse-stage formulation has four degrees of freedom ($(\tilde{\mathbf{p}}_{\rm{UE}}\in\mathbb{R}^2,\ b_{\rm{UE}}\in\mathbb{R},\ \alpha_{\rm{UE}}\in\mathbb{R})$). Since each measurement contributes one independent constraint direction near consistency, a minimal solver typically requires at least four measurements. 
\end{corollary}

Importantly, Corollaries 2 and 3 operate at the level of the proposed residual/Jacobian analysis and the corresponding initialization strategy. They do not replace classical geometric solvability results for snapshot radio SLAM under stronger path hypotheses. In particular, classical mixed-path facts, such as the solvability of the one-LoS-plus-one-NLoS-1 case under the corresponding geometric assumptions, remain valid and are not contradicted by the formulation-specific path-count discussion above.
\section{Simulation Results}\label{Results}

\begin{figure}
    \centering
    \begin{subfigure}[b]{0.35\textwidth}
        \centering
        \includegraphics[width=1\linewidth]{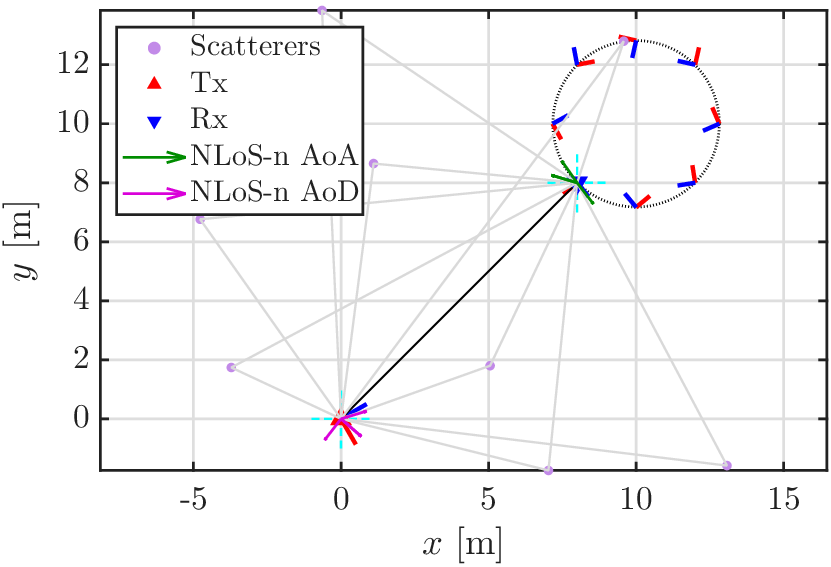}
        \caption{\small The planar simulation setup.}
        \label{simulation_trial}
    \end{subfigure}
    \hfill
    \begin{subfigure}[b]{0.49\textwidth}
        \centering
        \includegraphics[width=1\linewidth]{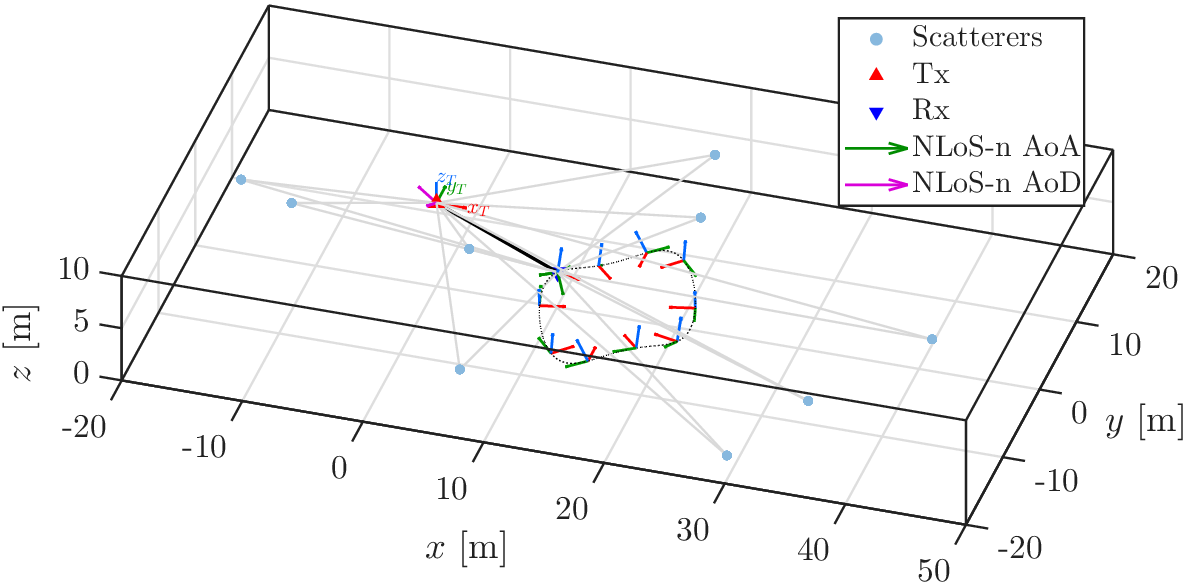}
        \caption{\small The 3-D simulation setup.}
        \label{simulation_trial_3D}
    \end{subfigure}
    \caption{\small The planar and 3-D simulation setup. In each trial, the NLoS-1 paths are generated with scattering points randomly distributed in the considered space, while the NLoS-n paths are characterized by uniformly distributed AoA and AoD angles.}
\end{figure}

\subsection{Simulation Setup}\label{Simulation Model and Data}

The proposed algorithm is evaluated using simulated measurements. In the planar scenario, the considered area is $20 \times 20$ m$^2$, with ${{\bf{p}}_{\rm{BS}}} = [0,0]$ m and ${\alpha _{\rm{BS}}} =  - \frac{\pi }{3}$ rad. The UE position is randomly selected from eight equally spaced points on a circle centered at $[10,10]$ m with radius $2\sqrt{2}$ m, and $\alpha_{\rm{UE}}$ is uniformly drawn from $[-\pi,\pi]$ rad.

For the general 3-D scenario, the considered space has size $70 \times 40 \times 10$ m$^3$, with ${{\bf{p}}_{{\rm{BS}}}} = [0,0,3]$ m and ${{\mathbf{R}}_{{\rm{BS}}}} = \mathbf{I}$. As shown in Fig. \ref{simulation_trial_3D}, eight UE states are sampled along a trajectory whose $XY$ projection is an ellipse centered at $[15,0,1]^{\operatorname{T}}$~m with semi-axes 8~m and 6~m, and whose height follows a sinusoid of amplitude 1~m over three periods. At each sampled location, the receiver frame is tangent to the trajectory.

For each accuracy level, 1000 Monte Carlo trials are generated. Each planar trial contains 8 inlier paths and 3 outlier paths, while each 3-D trial contains 9 inlier paths and 3 outlier paths. The proposed method is compared with three baselines where applicable: the MH-MLE approach in \cite{11046111}, the robust snapshot radio SLAM (RSR) method in \cite{10818978}, and the traditional snapshot radio SLAM method in \cite{9179819}. The MH-MLE and RSR baselines exploit calibrated path-loss or amplitude information for LoS/path-state handling, whereas the traditional method assumes NLoS-only propagation and does not perform LoS model selection. Since directly comparable published baselines for robust 3-D/6-D snapshot SLAM with unknown path identity are not readily available, external comparisons are conducted in the assumption-matched planar setting. The 3-D experiments validate the proposed full-pose formulation, threshold sensitivity, and final refinement.

For the $i$th path, the observation ${{\bf{z}}_i}$ is generated using \eqref{observe pdf}. The noise covariance matrix is
\begin{equation}\label{noise}
    {\mathbf{\Sigma }}_i = \gamma {\mathbf{\Sigma }}_{\ell^{(i)}}, \quad \ell^{(i)} \in \{\mathrm{LoS},\mathrm{NLoS}\text{-}1,\mathrm{NLoS}\text{-}n\},
\end{equation}
where ${\mathbf{\Sigma }}^{1/2}_{\mathrm{LoS}} = {\rm{diag}}([0.3,\tfrac{\pi}{180}\mathbf{1}_{1\times (k_d-1)}])$ \cite{10818978}, ${\mathbf{\Sigma }}^{1/2}_{\mathrm{NLoS}\text{-}n} = 2{\mathbf{\Sigma }}^{1/2}_{\mathrm{NLoS}\text{-}1}=4{\mathbf{\Sigma }}^{1/2}_{\mathrm{LoS}}$, and $\gamma  = [0.01,0.05,0.1,0.5,1]$ controls the measurement accuracy. We set ${b_{\rm{UE}}} = 10$ ns, $T_\epsilon = 1$, and $N_{\min} = 4$. For QAIC, $\hat c=\exp(\gamma)$ at the corresponding accuracy level. The planar angular resolution is $1^\circ$ ($|\mathcal{A}| = 360$), while the 3-D twist-swing traversal uses $\Delta_\psi=6^\circ$, $\Delta_\chi=12^\circ$, $\Delta_\beta=10^\circ$, and $\beta_{\max}=90^\circ$. Path means are generated according to their NLoS-1/NLoS-n labels, and path magnitudes follow the models in Appendix \ref{Appendix B}. One planar trial and one 3-D trial are shown in Figs. \ref{simulation_trial} and \ref{simulation_trial_3D}.

The evaluation metrics are defined consistently across the simulations. The UE position and clock-bias errors are $e_p=\|\hat{\mathbf p}_{\rm UE}-\mathbf p_{\rm UE}\|_2$ and $e_b=|\hat b_{\rm UE}-b_{\rm UE}|$. The orientation error is $e_R=|\operatorname{wrap}(\hat\alpha_{\rm UE}-\alpha_{\rm UE})|$ in 2-D and $e_R=\|\operatorname{Log}(\hat{\mathbf R}_{\rm UE}\mathbf R_{\rm UE}^{\top})^{\vee}\|_2$ in 3-D. Scatterer errors are computed after one-to-one nearest-neighbor matching between estimated and true NLoS-1 scatterers. The reported CDFs, P90 values, inlier F1 scores, and computation times are obtained over all Monte Carlo trials.

\subsection{Planar Simulation Results}

\subsubsection{Performance with Well-Calibrated Path Loss Model}\label{Simulation Results Analysis 2D}

\begin{figure*}[htbp] 
\centering
\begin{subfigure}[b]{1\textwidth}
\centering
\includegraphics[width=3.0in, height=0.2in]{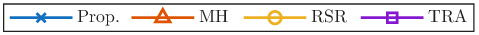} 
\end{subfigure}%
\vfill
\begin{subfigure}[b]{0.24\textwidth}
\centering
\includegraphics[width=1.6in, height=1.5in]{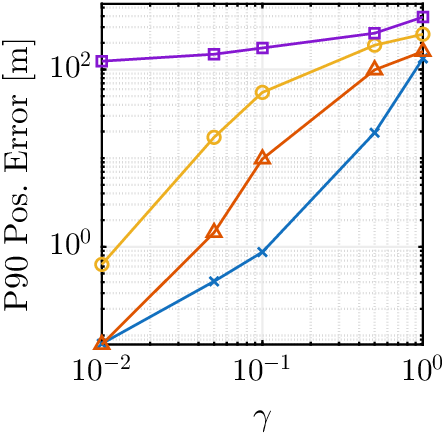} 
\subcaption{Position error}\label{pUE} 
\end{subfigure}%
\begin{subfigure}[b]{0.24\textwidth}
\centering
\includegraphics[width=1.6in, height=1.5in]{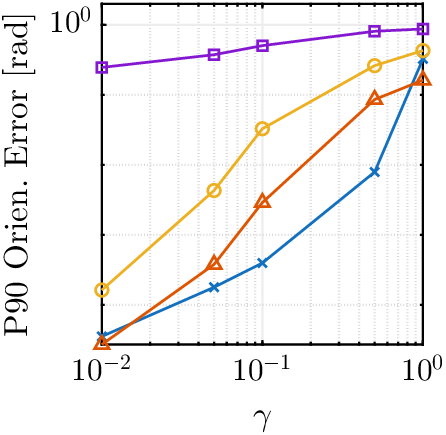}
\subcaption{Orientation error}\label{aUE} 
\end{subfigure}%
\begin{subfigure}[b]{0.24\textwidth}
\centering
\includegraphics[width=1.6in, height=1.5in]{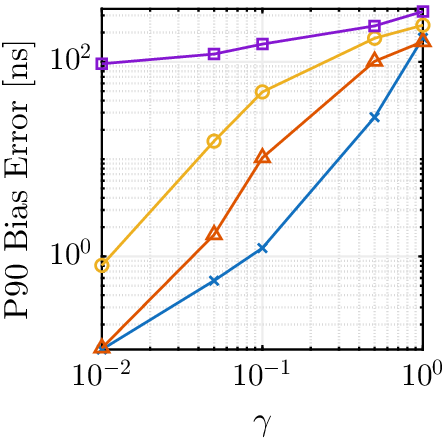}
\subcaption{Clock bias error}\label{bUE} 
\end{subfigure}
\begin{subfigure}[b]{0.24\textwidth}
\centering
\includegraphics[width=1.6in, height=1.5in]{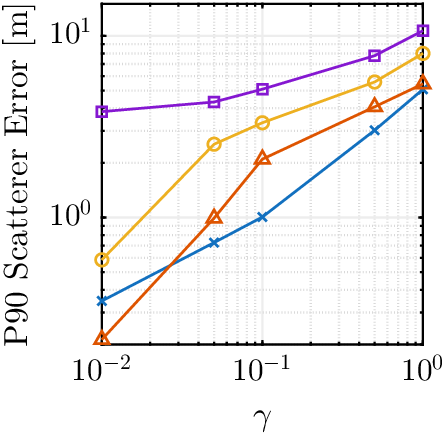}
\subcaption{Scatterer error}\label{sUE} 
\end{subfigure}
\caption{\small P90 errors of the considered methods in a planar scenario under different $\gamma$.} 
\label{LoS}
\end{figure*}

The empirical CDFs of the UE position, orientation, clock-bias, and scattering-point errors are evaluated in mixed conditions where the LoS path is present with probability $60\%$. The 90\% quantile (P90) is extracted for performance comparison. 
The methods listed in Section \ref{Simulation Model and Data} are compared under different noise levels $\gamma$, assuming that the path-loss model is well calibrated.  

Simulation results are illustrated in Fig. \ref{LoS}, where the figures give the P90 error versus $\gamma$. As expected, the P90 errors of all methods increase as $\gamma$ grows, since larger $\gamma$ corresponds to less accurate angle-delay observations. The proposed method achieves the lowest or comparable P90 errors in most cases, with a clear advantage in UE position, orientation, and clock-bias estimation over the low-to-moderate noise range. For scattering-point estimation, MH-MLE is slightly better at $\gamma=10^{-2}$, while the proposed method becomes more accurate as the measurement noise increases. This indicates that the unified angle-delay initialization provides a stable UE-state estimate, which in turn benefits the subsequent scattering-point refinement.

The traditional snapshot radio SLAM method exhibits the largest errors in most cases, which is consistent with its NLoS-only modeling assumption. In mixed propagation, forcing a LoS path into an NLoS-1 model introduces a structural mismatch and biases the UE-state estimate. The RSR and MH-MLE baselines exploit calibrated path-loss information for LoS handling and therefore improve over the traditional method under the well-calibrated setting, but their errors remain generally higher than those of the proposed method, especially as the measurement noise increases. The proposed method remains competitive or superior without amplitude-based LoS preclassification, because LoS identification is performed after robust amplitude-independent estimation rather than driven by a path-loss likelihood. This shows that the angle-delay consistency used in the coarse stage is sufficient to provide a reliable initialization and inlier set for the refinement stage.

\subsubsection{Performance in the Presence of Path Loss Model Uncertainty}\label{Fluctuation of path_loss_parameter_simulation}
\begin{figure*}[htbp] 
\centering
\begin{subfigure}[b]{0.24\textwidth}
\centering
\includegraphics[width=1.5in, height=1.3in]{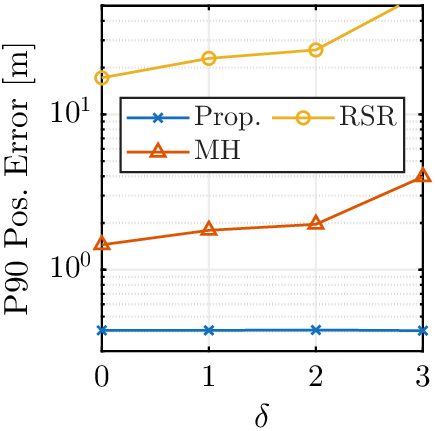} 
\subcaption{Position error}\label{pUE_flu} 
\end{subfigure}%
\begin{subfigure}[b]{0.24\textwidth}
\centering
\includegraphics[width=1.5in, height=1.3in]{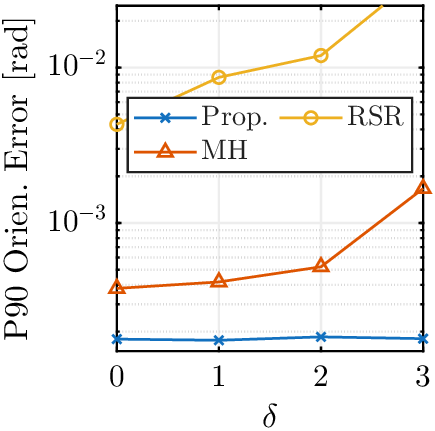}
\subcaption{Orientation error}\label{aUE_flu} 
\end{subfigure}
\begin{subfigure}[b]{0.24\textwidth}
\centering
\includegraphics[width=1.5in, height=1.3in]{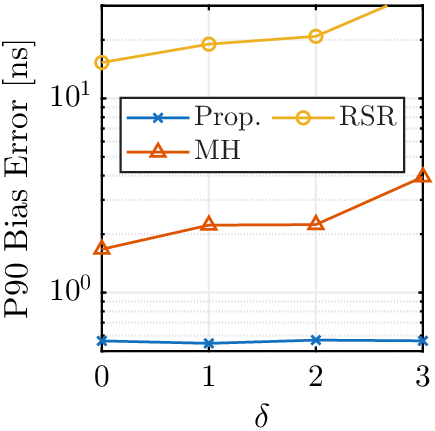}
\subcaption{Clock bias error}\label{bUE_flu}
\end{subfigure}
\begin{subfigure}[b]{0.24\textwidth}
\centering
\includegraphics[width=1.5in, height=1.3in]{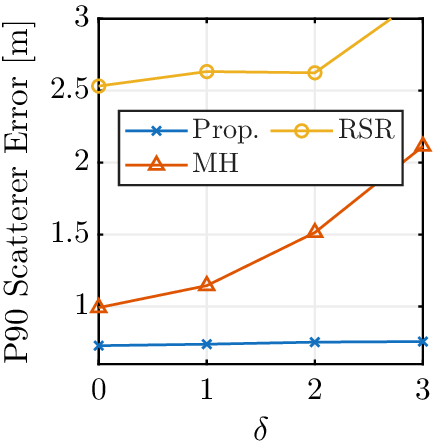}
\subcaption{Scatterer error}\label{sUE_flu}
\end{subfigure}
\caption{\small P90 errors of the considered methods in a planar scenario under different $\delta$.}
\label{uncertainty_pl}
\end{figure*}
This stress test is not intended to exhaust all robustness factors in practical channel estimation. Instead, it targets the calibration-mismatch mechanism that motivates the proposed amplitude-independent formulation.
To emulate calibration uncertainty, the path-loss parameters are modeled as random perturbations around their nominal values. Specifically, $L_0+\Delta L_0$, $\zeta+\Delta\zeta$, and $\sigma_{\rm{sh}}\exp(\Delta_{\sigma_{\rm{sh}}})$, where $\left[\Delta L_0,\Delta\zeta,\Delta_{\sigma_{\rm{sh}}}\right]^{\top}\sim\mathcal N(\mathbf{0},\mathbf{\Sigma}_{\rm pl})$ are sampled per trial. The perturbation levels are controlled by $\mathbf{\Sigma}_{\rm pl}^{1/2}=\rm{diag}\left(2^{\delta-1}\left[1,0.1,0.1\right]\right)$ with level factor $\delta$.
The simulations are conducted with $\gamma=10^{-1}$, where $\delta=0$ corresponds to the nominal path-loss model and larger $\delta$ indicates stronger calibration uncertainty. Fig. \ref{uncertainty_pl} compares the P90 errors of the proposed method and the amplitude-dependent baselines under different uncertainty levels. Since the proposed method does not use path-loss information for LoS preclassification or coarse-stage hypothesis generation, its curves remain nearly flat as $\delta$ increases. This behavior is observed consistently for UE position, orientation, clock bias, and scatterer estimation, confirming that the proposed angle-delay formulation is insensitive to the targeted path-loss-parameter mismatch.

In contrast, the performance of MH-MLE and RSR degrades when the path-loss model becomes less reliable. MH-MLE shows a gradual increase in all four error metrics, with the degradation becoming more evident at larger $\delta$. RSR is more sensitive in this experiment: its errors are already larger under the nominal model and increase further as the mismatch grows, especially for UE-state estimation. This is because perturbations in $L_0$, $\zeta$, and $\sigma_{\rm{sh}}$ distort the amplitude-based LoS likelihood and path-consistency evaluation, which can lead to incorrect path-identity decisions and poorer initialization. These results support the motivation of the proposed method: the robustness of amplitude-dependent methods deteriorates under the targeted path-loss-model calibration uncertainty, whereas the proposed amplitude-independent formulation preserves stable accuracy.

\subsubsection{LoS Detection Accuracy}\label{LoS Detection Accuracy}

\begin{table}[t]
\caption{LoS detection accuracy under different LoS rates}
\centering
\small
\setlength{\tabcolsep}{3.5pt}
\renewcommand{\arraystretch}{1.05}
\begin{tabular}{@{}c c c c c c c@{}}
\toprule
\multirow{2}{*}{$\bm{\gamma}$} 
  & \multicolumn{2}{c}{\textbf{70\% LoS Rate}} 
  & \multicolumn{2}{c}{\textbf{50\% LoS Rate}} 
  & \multicolumn{2}{c}{\textbf{30\% LoS Rate}} \\
\cmidrule(lr){2-3} \cmidrule(lr){4-5} \cmidrule(lr){6-7}
  & \textbf{Prop.} & \textbf{MH} 
  & \textbf{Prop.} & \textbf{MH} 
  & \textbf{Prop.} & \textbf{MH} \\
\midrule
$10^{-2}$ & 95.88\% & 89.01\% & 93.26\% & 85.06\% & 90.91\% & 80.95\% \\
$5\times 10^{-2}$ & 96.45\% & 88.50\% & 94.09\% & 84.33\% & 91.86\% & 80.04\% \\
$10^{-1}$  & 96.35\% & 88.20\% & 93.66\% & 83.81\% & 91.34\% & 79.25\% \\
$5\times 10^{-1}$  & 89.53\% & 86.43\% & 83.76\% & 80.73\% & 78.05\% & 74.78\% \\
$10^{0}$  & 84.70\% & 85.61\% & 77.16\% & 79.11\% & 69.44\% & 72.64\% \\
\bottomrule
\end{tabular}
\label{tab:los_detection}
\end{table}

Table \ref{tab:los_detection} compares the LoS detection performance of the different methods with an uncertainty factor of $\delta=3$. For $\gamma\le 5\times10^{-1}$, the proposed method consistently achieves higher LoS detection accuracy than MH-MLE under all tested LoS occurrence rates. This improvement is consistent with the design of the proposed refinement stage, where LoS identification is performed through information criteria after a robust amplitude-independent estimation, rather than relying on a potentially mismatched path-loss likelihood.

As $\gamma$ increases to $1$, the performance gap becomes small and MH-MLE is slightly better in all three LoS-rate cases. This indicates that when the geometric measurements are severely degraded, the advantage of amplitude-independent model selection is partially offset by the reduced reliability of the angle-delay residuals. Nevertheless, across most operating points in Table \ref{tab:los_detection}, especially in the relevant low-to-moderate noise regime, the proposed method provides more reliable LoS detection under path-loss-model uncertainty.

\subsubsection{Computation Time}\label{Computation Time}

\begin{table}[t]
\centering
\caption{Average runtime per planar trial}
\label{tab:Computation_Time}
\begin{tabular}{cccccc}
\toprule
\multirow{2}{*}{Case} & \multicolumn{3}{c}{\textbf{Prop.}} & \multirow{2}{*}{\textbf{MH}} & \multirow{2}{*}{\textbf{RSR}} \\ \cmidrule{2-4}
 & CF & LS & Total &  &  \\ \midrule
LoS & 0.5166s & 0.0114s  & 0.5280s & 0.2411s & 0.0966s \\
50\% LoS & 0.5152s & 0.0109s & 0.5260s & 0.5100s & 0.2283s \\
NLoS & 0.5141s & 0.0104s & 0.5245s & 0.7729s & 0.3582s \\ 
\bottomrule
\end{tabular}
\end{table}

The experiments are conducted on a desktop platform equipped with an Intel i9-14900K CPU and 64 GB of memory.
Table~\ref{tab:Computation_Time} gives the average computation time under different propagation conditions. The proposed method has a nearly constant runtime across the LoS, mixed, and NLoS cases, with a total computation time of $0.5245$--$0.5280$ s. This is because its robust coarse stage evaluates the same angle-delay consistency formulation regardless of whether the LoS path is present. The coarse stage (CF) dominates the cost, taking $0.5141$--$0.5166$ s, while the subsequent least-squares refinement (LS) only requires $0.0104$--$0.0114$ s. Hence, the refinement improves the final estimate with a small additional computational burden.

The baselines are faster in the pure LoS case, where reliable magnitude information reduces ambiguity, but their runtime increases as LoS information becomes less available. MH-MLE grows from $0.2411$ s to $0.7729$ s from LoS to NLoS, and RSR from $0.0966$ s to $0.3582$ s. Thus, the proposed method provides more stable cost and becomes faster than MH-MLE in the NLoS-dominant setting.

\subsection{General 3-D Simulation Results}\label{Performance Under Different Propagation Conditions for 3-D}

Since the external baselines are evaluated in the assumption-matched planar setting, the 3-D results examine whether the proposed full-pose pipeline remains stable under mixed propagation and how much the final LS refinement improves the coarse-formulation (CF) estimate.

\subsubsection{Sensitivity Evaluation}
\begin{figure}[t] 
\centering
\includegraphics[width=2.8in, height=1.4in]{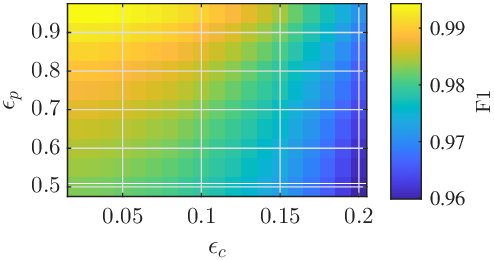}
\caption{\small 3-D inlier F1 sensitivity to feasibility thresholds $\varepsilon_c$ and $\varepsilon_p$, with $\gamma=0.1$.} 
\label{sen_pl}
\end{figure}

Fig. \ref{sen_pl} evaluates the sensitivity of the coarse inlier selection to the feasibility-check thresholds $\varepsilon_c$ and $\varepsilon_p$ in the 3-D scenario. The overall inlier F1 score remains high over the tested parameter range, staying above approximately $0.96$, which indicates that the proposed robust search is not overly sensitive to a single finely tuned threshold choice. The best performance appears in the region with smaller $\varepsilon_c$ and larger $\varepsilon_p$, where the LoS collinearity test is relatively strict and the NLoS-1 same-side consistency test requires stronger geometric agreement. In this region, the F1 score approaches $0.99$, showing that physically meaningful feasibility checks can effectively reject geometrically inconsistent hypotheses before the refinement stage.

As $\varepsilon_c$ increases, the F1 score gradually decreases, especially when $\varepsilon_p$ is small. This is expected because a larger collinearity threshold admits more ambiguous paths as LoS-consistent candidates, which can increase false inlier decisions in the presence of NLoS-n components. Conversely, increasing $\varepsilon_p$ improves the F1 score by making the NLoS-1 feasibility condition more selective. These results suggest that the feasibility check should be chosen moderately conservatively: strict enough to suppress physically inconsistent minimal-set hypotheses, but not so restrictive that valid noisy inliers are rejected. The smooth variation in Fig. \ref{sen_pl} also shows that the proposed pipeline has a broad stable operating region for practical threshold selection.

\subsubsection{Estimation Evaluation}

\begin{figure*}[htbp] 
\centering
\begin{subfigure}[b]{1\textwidth}
\centering
\includegraphics[width=5.0in, height=0.16in]{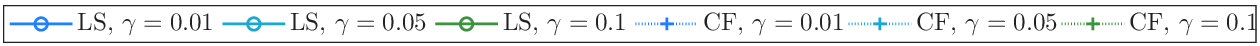} 
\end{subfigure}%
\vfill
\begin{subfigure}[b]{0.24\textwidth}
\centering
\includegraphics[width=1.6in, height=1.5in]{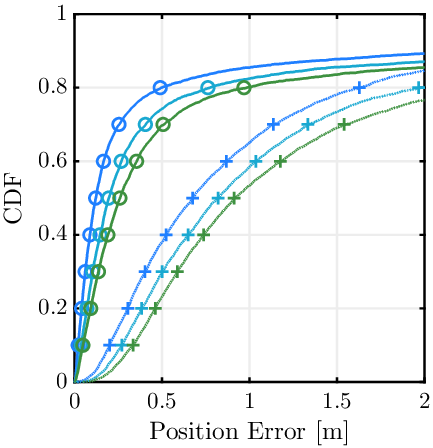} 
\subcaption{Position error}\label{pUE_3d} 
\end{subfigure}%
\begin{subfigure}[b]{0.24\textwidth}
\centering
\includegraphics[width=1.6in, height=1.5in]{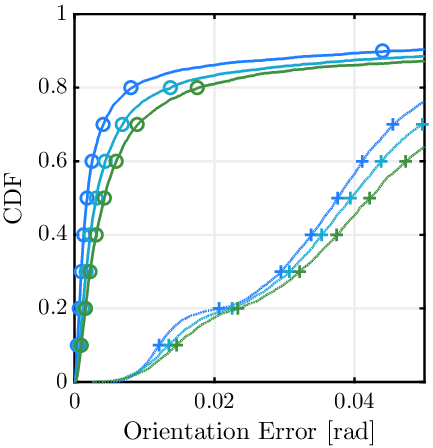}
\subcaption{Orientation error}\label{aUE_3d} 
\end{subfigure}%
\begin{subfigure}[b]{0.24\textwidth}
\centering
\includegraphics[width=1.6in, height=1.5in]{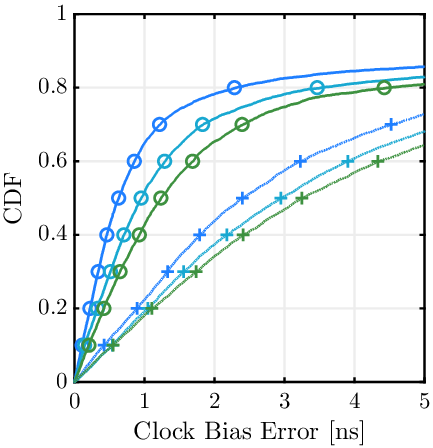}
\subcaption{Clock bias error}\label{bUE_3d} 
\end{subfigure}
\begin{subfigure}[b]{0.24\textwidth}
\centering
\includegraphics[width=1.6in, height=1.5in]{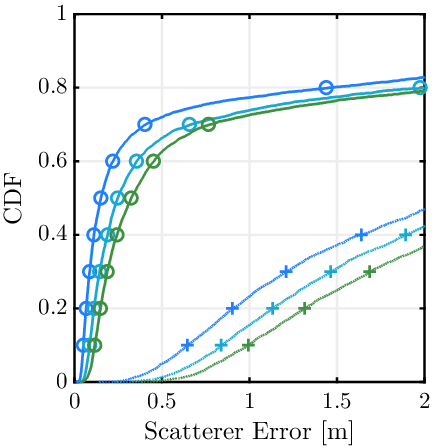}
\subcaption{Scatterer error}\label{sUE_3d} 
\end{subfigure}
\caption{\small Error CDFs of the proposed algorithm in the 3-D mixed-propagation scenario.} 
\label{LoS_3D}
\end{figure*}

The CDFs of the estimation errors are evaluated under mixed propagation conditions where the LoS path is present with probability $50\%$, with multiple noise levels configured ($\gamma  = [0.01,0.05,0.1]$). The evaluation includes the UE position, orientation, clock bias, and scattering-point errors. The metrics follow the definitions in Section \ref{Simulation Model and Data}.
The results are shown in Fig. \ref{LoS_3D}, where the coarse-formulation estimate (CF) and the final least-squares refinement (LS) are compared under the same mixed propagation setting. For all evaluated quantities, increasing $\gamma$ shifts the CDF curves to the right, which is consistent with the larger perturbations in the AoA, AoD, and delay measurements. Nevertheless, the curves remain ordered with respect to the noise level, showing that the proposed 3-D formulation degrades smoothly rather than failing abruptly as the measurement quality decreases.

Compared with the CF estimate, the LS refinement consistently shifts the error distributions to the left. For UE position and clock bias, the CF stage provides a usable initialization, while the LS stage further reduces the residual geometric mismatch. The improvement is especially visible for orientation and scattering-point estimation: the CF orientation errors remain broadly distributed, whereas the LS curves concentrate close to zero for all tested $\gamma$; similarly, the scatterer-error CDFs after LS refinement rise much faster than their CF counterparts. This indicates that, once the coarse stage has identified a reliable inlier set and an approximate UE state, fitting the physical LoS/NLoS-1 measurement model can jointly improve both the 6-D UE pose and the reconstructed landmarks.

These results verify the role of the two-stage design in the general 3-D scenario. The CF stage provides a robust amplitude-independent initialization under unknown path identity, and the LS stage uses this initialization to jointly refine the UE pose, clock bias, scattering points, and path interpretation. Hence, the proposed pipeline remains effective for full 6-D pose estimation and mapping, not only for the planar special case.

\section{Conclusion}\label{Conclusion}
This paper proposed an amplitude-independent robust snapshot radio SLAM method based on a unified angle-delay formulation for LoS and NLoS-1 inlier paths. The proposed coarse stage is written directly in the UE state and angle-delay observations, avoiding amplitude-based LoS preclassification and path-wise latent variables in the solver state. To handle unknown path identity and NLoS-n outliers, the coarse solver was embedded in a consensus-based minimal-set search with geometric feasibility checks. The formulation was further developed for general 3-D/6-D pose and clock-bias estimation, where the UE orientation is initialized by twist-swing two-stage traversal and refined locally on $SO(3)$. Building on the coarse estimate, a Jacobian-row-equilibrated IRLS refinement with QAIC-based model comparison jointly refines the UE state, detects the LoS path, and estimates scattering points. We also analyzed formulation-specific local-rank properties and derived minimal-set implications for the proposed initialization. Simulation results show that the proposed method remains competitive with calibrated amplitude-dependent baselines, provides robust LoS/inlier handling under the targeted path-loss-model mismatch, and remains effective in the general 3-D scenario.
\appendices
\section{Derivation of Derivatives for Iterative Estimation}\label{Appendix A}
If ${\bf{h}}_k^{(m)}$ is a LoS path in a hypothesis, we have $\nabla_{{{\bf{p}}_{{{\cal I}_m}}}}=\mathbf{0}_{1\times 3}$, ${\nabla _{{b_{{\rm{UE}}}}}}{\bf{h}}_k^{(m)} = [1,\mathbf{0}_{1\times 4}]^\top$, 
\begin{equation}
{\nabla _{{{\bf{p}}_{{\rm{UE}}}}}}{{\bf{h}}_{{\rm{LoS}}}} = \left[ {\begin{array}{*{20}{c}}
{\frac{1}{c}\frac{{{{\bf{t}}^ \top }}}{{\left\| {\bf{t}} \right\|}}}
\\
\frac{(\mathbf{R}_{\rm{BS}}^{\top}\mathbf{t})^{\top}[-\mathbf{e}_2,\mathbf{e}_1,\mathbf{0}]\mathbf{R}_{\rm{BS}}^{\top}}
{\|[\mathbf{e}_1,\mathbf{e}_2]^ \top \mathbf{R}_{\rm{BS}}^{\top}\mathbf{t}\|^2}
\\
\frac{{\bf{e}}_3^ \top {\bf{R}}_{{\rm{BS}}}^ \top \left[\|{\bf{t}}\|^2\mathbf{I}-{\bf{t}}{\bf{t}}^{\top}\right] }{\left\|[\mathbf{e}_1,\mathbf{e}_2]^ \top \mathbf{R}_{\rm{BS}}^{\top}\mathbf{t}\right\|\left\| {{\bf{t}}} \right\|^2}
\\
\frac{(\mathbf{R}_{\rm{UE}}^{\top}\mathbf{t})^{\top}[\mathbf{e}_2,-\mathbf{e}_1,\mathbf{0}]\mathbf{R}_{\rm{UE}}^{\top}}
{\|[\mathbf{e}_1,\mathbf{e}_2]^ \top \mathbf{R}_{\rm{UE}}^{\top}\mathbf{t}\|^2}
\\
-\frac{{\bf{e}}_3^ \top {\bf{R}}_{{\rm{UE}}}^ \top \left[\|{\bf{t}}\|^2\mathbf{I}-{\bf{t}}{\bf{t}}^{\top}\right] }{\left\|[\mathbf{e}_1,\mathbf{e}_2]^ \top \mathbf{R}_{\rm{UE}}^{\top}\mathbf{t}\right\|\left\| {{\bf{t}}} \right\|^2}
\end{array}} \right],
\end{equation}
and
\begin{equation}
{\nabla _{{{\bf{R}}_{{\rm{UE}}}}}}{{\bf{h}}_{{\rm{LoS}}}} = \left[ {\begin{array}{*{20}{c}}
{{{\bf{0}}_{3 \times 3}}}
\\
\frac{(\mathbf{R}_{\rm{UE}}^{\top}\mathbf{t})^{\top}[-\mathbf{e}_2,\mathbf{e}_1,\mathbf{0}]\mathbf{R}_{\rm{UE}}^{\top}\left[\mathbf{t}\right]_{\times}}
{\|[\mathbf{e}_1,\mathbf{e}_2]^ \top \mathbf{R}_{\rm{UE}}^{\top}\mathbf{t}\|^2}
\\
-\frac{\|\mathbf{t}\|{\bf{e}}_3^ \top {\bf{R}}_{{\rm{UE}}}^ \top \left[\mathbf{t}\right]_{\times}}{\left\|[\mathbf{e}_1,\mathbf{e}_2]^ \top \mathbf{R}_{\rm{UE}}^{\top}\mathbf{t}\right\|}
\end{array}} \right],
\end{equation}
Otherwise, if ${\bf{h}}_k^{(m)}$ is a NLoS-1 path in a hypothesis, we have ${\nabla _{{b_{{\rm{UE}}}}}}{\bf{h}}_k^{(m)} = [1,\mathbf{0}_{1\times 4}]^\top$, 
\begin{equation}
{\nabla _{{{\bf{p}}_{{{\cal I}_m}}}}}{{\bf{h}}_{{\rm{NLoS - 1}}}} = 
\begin{bmatrix}
\frac{1}{c}{{\left( {\tfrac{{{\bf{t}}_{{{\cal I}_m}}^ \vee }}{{\left\| {{\bf{t}}_{{{\cal I}_m}}^ \vee } \right\|}} + \tfrac{{{\bf{t}}_{{{\cal I}_m}}^ \wedge }}{{\left\| {{\bf{t}}_{{{\cal I}_m}}^ \wedge } \right\|}}} \right)}^ \top }
\\
\frac{(\mathbf{R}_{\rm{BS}}^{\top}\mathbf{t}_{\mathcal{I}_m}^{\wedge})^{\top}[-\mathbf{e}_2,\mathbf{e}_1,\mathbf{0}]\mathbf{R}_{\rm{BS}}^{\top}}
{\|[\mathbf{e}_1,\mathbf{e}_2]^ \top \mathbf{R}_{\rm{BS}}^{\top}\mathbf{t}_{\mathcal{I}_m}^{\wedge}\|^2}
\\
\frac{{\bf{e}}_3^ \top {\bf{R}}_{{\rm{BS}}}^ \top \left[\|{\bf{t}}_{{{\cal I}_m}}^ \wedge\|^2\mathbf{I}-({\bf{t}}_{{{\cal I}_m}}^ \wedge)({\bf{t}}_{{{\cal I}_m}}^ \wedge)^{\top}\right] }{\left\|[\mathbf{e}_1,\mathbf{e}_2]^ \top \mathbf{R}_{\rm{BS}}^{\top}\mathbf{t}_{\mathcal{I}_m}^{\wedge}\right\|\left\| {{\bf{t}}_{{{\cal I}_m}}^ \wedge } \right\|^2} 
\\
\frac{(\mathbf{R}_{\rm{UE}}^{\top}\mathbf{t}_{\mathcal{I}_m}^{\vee})^{\top}[-\mathbf{e}_2,\mathbf{e}_1,\mathbf{0}]\mathbf{R}_{\rm{UE}}^{\top}}
{\|[\mathbf{e}_1,\mathbf{e}_2]^ \top \mathbf{R}_{\rm{UE}}^{\top}\mathbf{t}_{\mathcal{I}_m}^{\vee}\|^2}
\\
-\frac{{\bf{e}}_3^ \top {\bf{R}}_{{\rm{UE}}}^ \top \left[\|{\bf{t}}_{{{\cal I}_m}}^ \vee\|^2\mathbf{I}-({\bf{t}}_{{{\cal I}_m}}^ \vee)({\bf{t}}_{{{\cal I}_m}}^ \vee)^{\top}\right] }{\left\|[\mathbf{e}_1,\mathbf{e}_2]^ \top \mathbf{R}_{\rm{UE}}^{\top}\mathbf{t}_{\mathcal{I}_m}^{\vee}\right\|\left\| {{\bf{t}}_{{{\cal I}_m}}^ \vee } \right\|^2}
\end{bmatrix},
\end{equation}
\begin{equation}
{\nabla _{{{\bf{p}}_{{\rm{UE}}}}}}{{\bf{h}}_{{\rm{NLoS - 1}}}} = 
\begin{bmatrix}
 - \frac{1}{c}\frac{{{{\left( {{\bf{t}}_{{{\cal I}_m}}^ \vee } \right)}^ \top }}}{{\|{\bf{t}}_{{{\cal I}_m}}^ \vee }\|}
\\
{{{\bf{0}}_{2 \times 3}}}
\\
\frac{(\mathbf{R}_{\rm{UE}}^{\top}\mathbf{t}_{\mathcal{I}_m}^{\vee})^{\top}[\mathbf{e}_2,-\mathbf{e}_1,\mathbf{0}]\mathbf{R}_{\rm{UE}}^{\top}}
{\|[\mathbf{e}_1,\mathbf{e}_2]^ \top \mathbf{R}_{\rm{UE}}^{\top}\mathbf{t}_{\mathcal{I}_m}^{\vee}\|^2}
\\
-\frac{{\bf{e}}_3^ \top {\bf{R}}_{{\rm{UE}}}^ \top \left[\|{\bf{t}}_{{{\cal I}_m}}^ \vee\|^2\mathbf{I}-({\bf{t}}_{{{\cal I}_m}}^ \vee)({\bf{t}}_{{{\cal I}_m}}^ \vee)^{\top}\right] }{\left\|[\mathbf{e}_1,\mathbf{e}_2]^ \top \mathbf{R}_{\rm{UE}}^{\top}\mathbf{t}_{\mathcal{I}_m}^{\vee}\right\|\left\| {{\bf{t}}_{{{\cal I}_m}}^ \vee } \right\|^2}
\end{bmatrix},
\end{equation}
and
\begin{equation}
{\nabla _{{{\bf{R}}_{{\rm{UE}}}}}}{{\bf{h}}_{{\rm{NLoS - 1}}}} = 
\begin{bmatrix}
{{{\bf{0}}_{3 \times 3}}}
\\
\frac{(\mathbf{R}_{\rm{UE}}^{\top}\mathbf{t}_{\mathcal{I}_m}^{\vee})^{\top}[-\mathbf{e}_2,\mathbf{e}_1,\mathbf{0}]\mathbf{R}_{\rm{UE}}^{\top}\left[\mathbf{t}_{\mathcal{I}_m}^{\vee}\right]_{\times}}
{\|[\mathbf{e}_1,\mathbf{e}_2]^ \top \mathbf{R}_{\rm{UE}}^{\top}\mathbf{t}_{\mathcal{I}_m}^{\vee}\|^2}
\\
-\frac{\|\mathbf{t}_{\mathcal{I}_m}^{\vee}\|{\bf{e}}_3^ \top {\bf{R}}_{{\rm{UE}}}^ \top \left[\mathbf{t}_{\mathcal{I}_m}^{\vee}\right]_{\times}}{\left\|[\mathbf{e}_1,\mathbf{e}_2]^ \top \mathbf{R}_{\rm{UE}}^{\top}\mathbf{t}_{\mathcal{I}_m}^{\vee}\right\|}
\end{bmatrix}.
\end{equation}
where ${\bf{t}}_{\mathcal{I}_m}^{\vee} = ({{\bf{p}}_{{{\cal I}_m}}} - {{\bf{p}}_{{\rm{UE}}}})$ and ${\bf{t}}_{\mathcal{I}_m}^{\wedge} = ({{\bf{p}}_{{{\cal I}_m}}} - {{\bf{p}}_{{\rm{BS}}}})$.

Further, in a planar scenario, the Jacobian matrix degenerates into $\left[\nabla_{\mathbf{\Theta}^{(m)}}\mathbf{h}^{(m)}_k\right]_{(\mathbf{s}_r,\mathbf{s}_c)}$, where the operator $\left[\mathbf{A}\right]_{(\mathbf{s}_r,\mathbf{s}_c)}$ returns a submatrix formed by the rows and columns of $\mathbf{A}$ that are selected by the vectors $\mathbf{s}_r=[1,2,4]^\top$ and $\mathbf{s}_c=[3]$, respectively.
\section{Generation of NLoS Paths}\label{Appendix B}
The NLoS-1 path is generated with its scattering point $\mathbf{p}_i$, which is randomly generated within the space. Consequently, the angle and path delay are computed using \eqref{observe mean}.
On the other hand, the NLoS-n path is generated with uniformly distributed angles, that is,
\begin{equation}\label{uniform1}
    {\phi _{\rm{NLoS - n},a}},{\theta _{\rm{NLoS - n},a}} \sim \mathcal{U}(0,2\pi ),
\end{equation}
\begin{equation}\label{uniform2}
    {\phi _{\rm{NLoS - n},e}},{\theta _{\rm{NLoS - n},e}} \sim \mathcal{U}( - \tfrac{\pi }{2},\tfrac{\pi }{2}).
\end{equation}
Since an NLoS-$n$ path may result from multiple bounces, it can exhibit a longer propagation distance and broader delay dispersion. To emulate this phenomenon in a synthetic outlier-delay model, we use a skew-normal random variable \cite{skewnormal} to construct the NLoS-$n$ path delay, i.e., ${\tau _{\rm{NLoS - n}}} \sim \mathcal{SN}(\mu ,\sigma ,\alpha)$, 
where $\alpha$ is the shape parameter, $\mu$ and $\sigma$ specify the resulting mean and standard deviation. In our simulation setup, we set $\mu$ to the time-delay mean of all inliers, choose $\sigma=100$ ns and $\alpha=-4$ to introduce asymmetry in the synthetic NLoS-$n$ outlier-delay model.  

The amplitude $\xi_i = \sqrt{\eta_i}$ of each NLoS path (either NLoS-1 or NLoS-n) is generated with total NLoS power and the exponential decay model as $
{\eta_{i}} = \tfrac{{\eta'_{i}}}{{\sum\nolimits_{k = 1}^{N'} {\eta_{k}^{'}} }} \times {\eta_{\rm{NLoS}}}$,
where ${\eta'_{i}} = 10^{ \frac{{S}}{10 }}{e^{ - \frac{{{\tau _i}}}{\beta }}}$, $S$ is the per-path shadowing and $S \sim{\cal N}({\mu _S},{\sigma_{S}^2})$, $N'$ denotes the total number of NLoS paths, and $\eta_{\rm{NLoS}}$ represents the total NLoS power. We set $\beta =31.4$ ns, $\mu_S=0$ dB and $\sigma_S=9.4$ dB \cite{7037347} in the simulation setup.

In the case of the LoS scenario, the total NLoS power $\eta_{\rm{NLoS}}$ is specified by ${\eta_{\rm{NLoS}}} = \frac{{\eta_{\rm{LoS}}}}{10^{K/10}}$,
where the K-factor is assumed $K \sim \mathcal{N}({\mu_{k}},{\sigma_{k}^2})$, and we set $\mu_k=7$ dB and $\sigma_k=4$ dB following \cite{zhu20213gpp}. 
For the power ${\eta^{\rm{dB}}}$, we use the logarithmic distance path loss model, that is,
${\eta^{\rm{dB}}} \sim \mathcal{N}(f({{\bf{p}}_{\rm{UE}}}),{\sigma_{\rm{sh}} ^{2}})$,
where $f({{\bf{p}}_{\rm{UE}}}){\rm{ }} = -{L_0} - {\rm{ }}10\zeta \log_{10}(\| {{\bf{p}}_{\rm{BS}}} - {{\bf{p}}_{\rm{UE}}}\|)$, $L_0$ is the reference path loss at a distance of one meter and $\zeta$ is the path loss coefficient. 
In LoS scenario, we use the path loss coefficient $\zeta= 1.7$, standard deviation $\sigma_{\rm{sh}}=1.8$ dB, and the $L_0 = $ 13 dB.
Otherwise, in the case of the NLoS scenario, we set $L_0 = $ 13 dB, $\zeta=3.3$ dB, and $\sigma_{\rm{sh}}= 4.6$ dB \cite{10818978}, \cite{8496691}.

\bibliographystyle{IEEEtran}
\bibliography{mylib}

\end{document}